\documentclass[12pt,notitlepage,tightenlines,superscriptaddress,onecolumn,twoside]{revtex4-1}

\usepackage{amsmath,amssymb,amsthm}
\usepackage{amsfonts}
\usepackage[T1]{fontenc}

\usepackage{graphicx}
\usepackage{epstopdf}

\usepackage{fancyhdr}
\usepackage{subfigure}
\usepackage{MnSymbol}

\usepackage{hyperref}
\usepackage[all]{hypcap}

\hypersetup{
	pdftitle={},
	pdfauthor={},
	colorlinks=true,
	linkcolor=black,
	citecolor=black,
	filecolor=black,
	urlcolor=black,
	bookmarksopen,
	bookmarksopenlevel=1,
}

\let\a=\alpha \let\b=\beta    \let\g=\gamma  \let\d=\delta     \let\e=\varepsilon
  \let\h=\eta     \let\th=\theta     \let\l=\lambda
\let\m=\mu                       \let\r=\rho
\let\s=\sigma       \let\ph=\varphi

\def\ee{{\underline \varepsilon}}

\let\==\equiv

\let\io=\infty
\let\0=\noindent

\def\tm{\tilde{\m}}
\def\tp{\tilde{p}}
\def\tf{\tilde{f}}
\def\Hint{H'}
\def\Vint{V}
\def\Pint{P}
\def\Wint{W}
\def\vF{v_F}
\DeclareMathOperator{\sech}{sech}

\newcommand{\ppr}{^{\vphantom{\prime}}}
\newcommand{\wick}[1]{\left. :\! \hspace{-0.5pt} #1 \hspace{-0.5pt} \!: \right.}
\newcommand{\xxa}{\stackrel{\scriptscriptstyle \times}{\scriptscriptstyle \times}\!}
\newcommand{\xxe}{\!\stackrel{\scriptscriptstyle \times}{\scriptscriptstyle \times}}
\newcommand{\bosonnoord}[1]{\left. \xxa #1 \xxe \right.}

\def\tg{{\tilde g}}

\def\be{\begin{equation}}    \def\ee{\end{equation}}
\def\bea{\begin{eqnarray}}   \def\eea{\end{eqnarray}}
\def\bean{\begin{eqnarray*}} \def\eean{\end{eqnarray*}}
\def\bfr{\begin{flushright}} \def\efr{\end{flushright}}
\def\bc{\begin{center}}      \def\ec{\end{center}}
\def\bal{\begin{align}}      \def\eal{\end{align}}
\def\ba#1{\begin{array}{#1}} \def\ea{\end{array}}
\def\bd{\begin{description}} \def\ed{\end{description}}
\def\bv{\begin{verbatim}}    \def\ev{\end{verbatim}}

\newtheorem{theorem}{Theorem}[section]
\newtheorem{lemma}{Lemma}[section]

\newtheorem{oss}[theorem]{Remark}

\begin{document}


\title{Steady states and universal conductance in a quenched Luttinger model}

\author{Edwin Langmann}
\email[Electronic address: ]{langmann@kth.se}
\affiliation{Department of Theoretical Physics, KTH Royal Institute of Technology, 106 91 Stockholm, Sweden}

\author{Joel L. Lebowitz}
\email[Electronic address: ]{lebowitz@math.rutgers.edu}
\affiliation{Departments of Mathematics and Physics, Rutgers University, Piscataway, New Jersey 08854, USA}

\author{Vieri Mastropietro}
\email[Electronic address: ]{vieri.mastropietro@unimi.it}
\affiliation{Dipartimento di Matematica, Universit{\`a} degli Studi di Milano, 20133 Milano, Italy}

\author{Per Moosavi}
\email[Electronic address: ]{pmoosavi@kth.se}
\affiliation{Department of Theoretical Physics, KTH Royal Institute of Technology, 106 91 Stockholm, Sweden}

\date{May 30, 2016}

\begin{abstract}
We obtain exact analytical results for the evolution of a 1+1-dimensional Luttinger model prepared in a domain wall initial state, i.e., a state with different densities on its left and right sides.
Such an initial state is modeled as the ground state of a translation invariant Luttinger Hamiltonian $H_{\lambda}$ with short range non-local interaction and different chemical potentials to the left and right of the origin.
The system evolves for time $t>0$ via a Hamiltonian $H_{\lambda'}$ which differs from $H_{\lambda}$ by the strength of the interaction.
Asymptotically in time, as $t \to \infty$, after taking the thermodynamic limit, the system approaches a translation invariant steady state.
This final steady state carries a current $I$ and has an effective chemical potential difference $\mu_+ - \mu_-$ between right- ($+$) and left- ($-$) moving fermions obtained from the two-point correlation function.
Both $I$ and $\mu_+ - \mu_-$ depend on $\lambda$ and $\lambda'$.
Only for the case $\lambda = \lambda' = 0$ does $\mu_+ - \mu_-$ equal the difference in the initial left and right chemical potentials.
Nevertheless, the Landauer conductance for the final state, $G = I/(\mu_+ - \mu_-)$, has a universal value equal to the conductance quantum $e^2/h$ for the spinless case.
\end{abstract}

\maketitle


\pagestyle{fancy}
\fancyhf{}
\fancyhead[RE, LO]{\nouppercase\leftmark}
\fancyhead[RO, LE]{\thepage}
\renewcommand{\headrulewidth}{0pt}
\renewcommand{\footrulewidth}{0pt}


\renewcommand{\thesection}{\arabic{section}}


\section{Introduction}
\label{Sec:Introduction}
The transport properties of a mesoscopic system are manifested in the evolution of its locally conserved quantities, such as particle and energy densities, following a quench from a non-uniform state.
In its simplest form, one prepares an isolated system in an initial state at time $t = 0$ with different density or temperature profiles to the left and right of the system, and then lets it evolve according to its internal translation invariant Hamiltonian.
The state of the system at a time $t > 0$ will then depend on the initial state and on the nature of the Hamiltonian.
After a long time, a system with ``good'' ergodic properties will forget the details of its initial state and come to \emph{thermal equilibrium} depending only on the total energy and on the number of particles of the initial state.
This is what is expected to be true for typical quantum systems.
The exceptions are integrable systems, in which there are many conserved quantities, and systems with many-body localization; see, e.g., \cite{EFG, PSSV, GHLT, GME, RDYO}.
For such systems there will still be an approach to some form of steady state, and this is sometimes called {\it equilibration} or {\it stabilization}.
(While we would prefer the latter terminology to distinguish from \emph{thermal equilibration}, also known as \emph{thermalization}, the former is used for reasons of convention.)

The time of approach to a steady state will depend on the size of the system.
Moreover, if one considers the entire system, there can never be a full loss of memory of the initial state.
After all, the evolution of an isolated quantum system is reversible, and for a finite quantum system the evolution is even quasi-periodic.
The approach to a steady state is therefore to be taken in a weak sense, i.e., one has to look at local (coarse-grained) quantities and wait a long time but not too long.
This should be equivalent to start with a finite system, say a one-dimensional system on the interval $[-L/2, L/2]$ for $L > 0$ prepared in an initial state which is different to the left and right of the origin, evolve the system for a time $t$, and consider its state projected on a subsystem taken as the interval $[-\ell, \ell]$ for $L > \ell > 0$, followed by first letting $L \to \io$ and then $t \to \io$ while keeping $\ell$ fixed but arbitrary.
In this way we can expect to obtain a steady state, described by a density matrix on the interval $[-\ell, \ell]$, and then ask for the density profile and the current in this final steady state \cite{SL, JP1, JP2, AJPP}.

To anticipate the properties of a final steady state in this set-up we consider first a system of length $L > 0$ in contact with infinite reservoirs at its left and right boundaries with different fixed chemical potentials or temperatures $(\m_L, T_L)$ and $(\m_R, T_R)$, respectively.
The coupling between system and reservoirs is done stochastically for classical systems \cite{RLL, Dh, BLRb, LLP, BO, BOS, RLL}, and for quantum systems one uses Lindblad-type operators \cite{JP1, JP2, AJPP, Li}.
Such a system will in general approach a steady state for fixed $L$.
We can then define the \emph{electrical conductivity} $\s$ as the ratio of the steady particle current $I$ to the average gradient $(\m_L-\m_R)/L$:
\be
\s = \frac{I}{(\m_L-\m_R)/L}.
\label{df}
\ee
In the same way the \emph{thermal conductivity} is defined as the ratio of the steady heat current to $(T_L-T_R)/L$.
In general, $\s \sim L^\a$ as $L \to \io$, with $\a$ depending on the type of system \cite{RLL, Dh, BLRb, LLP, BO, BOS, LMP, AJPP, AjHu}.
Several cases can be distinguished:
\begin{enumerate}

\item $\a = 0$ (normal conductivity).
In this case the system obeys \emph{Fourier's law}.
This has been shown, so far, only for classical systems with non-momentum conserving stochasticity in the bulk dynamics \cite{BO, BOS}.
\label{case_normal_conductivity}

\item $\a = 1$ (perfect conductivity).
This is the case, e.g., for fully integrable systems such as the harmonic crystal \cite{RLL} or free fermions \cite{AJPP}.
These systems have freely moving particles that carry the current.
\label{case_perfect_conductivity}

\item  $0 < \a < 1$ (enhanced conductivity).
This is expected to be the case for anharmonic one-dimensional systems with momentum conserving interactions as, e.g., in the Fermi-Pasta-Ulam system.
The only case where this has been proven rigorously, with $\a = 1/2$, is for a classical harmonic chain with momentum conserving stochastic interactions \cite{LMP}.
\label{case_enhanced_conductivity}

\item  $\a = -\io$ (zero conductivity). More precisely, $\s \to 0$ exponentially in $L$.
This is the case when one has localization as in a harmonic chain with random pinnings \cite{AjHu}.
Here too the result is for a classical system coupled to the reservoirs via Langevin terms.
We expect, however, this to be the same for quantum systems.
\label{case_zero_conductivity}

\end{enumerate}
We note that in cases \ref{case_normal_conductivity} and \ref{case_enhanced_conductivity} it is not clear a priori how to model the stochastic interactions for a quantum system.
This is an important open problem.

Going back to the isolated system of interest here, we expect the behavior of the subsystem on $[-\ell, \ell]$ in the limit $L \to \io$ followed by the limit $t \to \io$ to be as follows.
In case~\ref{case_normal_conductivity}, the subsystem will be in a thermal equilibrium state with vanishing current.
In case~\ref{case_perfect_conductivity}, it will be in a steady state that is translation invariant and has a non-vanishing current.
In case~\ref{case_enhanced_conductivity}, it will again be in a translation invariant steady state but without any current.
In case~\ref{case_zero_conductivity}, it will be in a steady state that maintains the initial profile, or something close to this, and thus will not be translation invariant or have a current.

Case~\ref{case_perfect_conductivity} in the above general classification can be checked in concrete quantum models, simple enough to be accessible by analytical or numerical methods.
Examples include quantum XX spin chains \cite{ARRS, RDYO} describing free fermions; see also \cite{ABGM2} for closely related work. 
In this case, the absence of interaction makes the system solvable, and, starting from a domain wall density profile, one gets a final steady state carrying a non-vanishing current.
The XXZ model is an extension of the XX model describing interacting fermions, which is exactly solvable by Bethe ansatz.
However, despite interesting recent progress \cite{LN}, it still seems unclear if this solution can be used to acquire full information on the evolution of domain walls.
Indeed, existing results on the evolution of quantum XXZ spins chains from a domain wall state are mostly numerical or based on approximations \cite{SM, LaMi}.

Here we shall consider a model for interacting (spinless) fermions that is more accessible to an analytical investigation, namely the Luttinger model \cite{Lu,To,ML} with a non-local interaction.
The first correct solution of this model and its ground state two-point correlation function were obtained in \cite{ML}, and the evolution after a quench from a homogenous state was studied in \cite{C, IC} for a local interaction and in \cite{MW} for a non-local interaction.
We also note that, concerning its equilibrium properties, the Luttinger model is a prototype for a larger equivalence class of one-dimensional systems called Luttinger liquids \cite{Ha}.

The Luttinger Hamiltonian
\begin{multline}
H_{\l} = \sum_{r=\pm} \int_{-L/2}^{L/2} dx\,
	\! \wick{ \tilde\psi^+_{r}(x)
		\left( -ir \vF \partial_x
			- \m_0 \left( 1 + \frac{2\l}{\pi\vF} \int_{-{L/2}}^{L/2} dy \Vint(y) \right)
		\right) \tilde\psi^-_{r}(x) } \\
+ \l \sum_{r,r'=\pm} \int_{-{L/2}}^{L/2} dx\, dy\, \Vint(x-y)
	\left(
	\! \wick{ \tilde\psi^+_{r\ppr}(x) \tilde\psi^-_{r\ppr}(x) } \!
	\! \wick{ \tilde\psi^+_{r'}(y) \tilde\psi^-_{r'}(y) } \!
	- \left( \frac{\m_0}{2\pi\vF} \right)^2
	\right)
\label{hamiltonian}
\end{multline}
describes right- and left-moving fermions on a line for $r = +$ and $r = -$, respectively, given~by fields $\tilde\psi^-_{r}(x)$ and $\tilde\psi^+_{r}(x) = \tilde\psi^-_{r}(x)^{\dagger}$ satisfying antiperiodic boundary conditions and the canonical anticommutation relations
\be
\left\{ \tilde\psi^-_{r}(x), \tilde\psi^+_{r'}(x') \right\}
	= \d_{r,r'} \d(x-x'),
\qquad
\left\{ \tilde\psi^\pm_{r}(x), \tilde\psi^\pm_{r'}(x') \right\}
	= 0,
\label{psi_tilde_CAR}
\ee
with $\wick{\cdots}$ indicating \emph{Wick} (\emph{normal}) \emph{ordering}, \emph{Fermi velocity} $v_F$, \emph{chemical potential} $\m_0$, \emph{coupling constant} $\l$, and a short range non-local \emph{interaction potential} $V(x-y)$; see, e.g., \cite{V, BGM, LaMo} and references therein.
The chemical potential $\m_0$ corresponds to the filled Dirac sea, and we adopt the description where this is taken as the ground state.
This, however, means that there can be both positive and negative densities, which should be interpreted as relative densities to a large constant ground state density.
We consider an initial state with different density profiles to the left and right of the system, modeled as the ground state $|\Psi_{\l,\m}\rangle$ of the Hamiltonian
\be
H_{\l,\m} = H_{\l} - \sum_{r=\pm} \int_{-L/2}^{L/2} dx\, ( \m_L \th(-x) + \m_R \th(x) - \m_0 )
\left( \wick{ \tilde\psi^+_{r}(x) \tilde\psi^-_{r}(x) } - \frac{\m_0}{2\pi\vF} \right)
\label{hamiltonian_with_muL_and_muR}
\ee
with different chemical potentials to the left, $\m_L = \m_0 + \m/2$, and right, $\m_R = \m_0 -\m/2$, assuming for definiteness $\m_L > \m_R$, i.e., $\m > 0$ as illustrated in Fig.~\ref{Fig:External_field}.
We consider the evolution of this state under a Luttinger Hamiltonian $H_{\l'}$ with a new coupling constant $\l'$, i.e., we consider the state $|\Psi_{\l,\m}^{\l'}(t)\rangle = e^{-iH_{\l'}t}|\Psi_{\l,\m}\rangle$.
In the case $\l = \l'$ this corresponds to an experiment in which one switches off an external field producing an excess of density on one side of the system compared to the other at time $t = 0$ and considers the evolution of the system under the translation invariant Hamiltonian $H_{\l}$ for time $t > 0$.
On the other hand, if $\l \neq \l'$, there is in addition an interaction quench, i.e., at $t=0$ there is also a change from $\l$ to $\l'$.

Let us first consider the non-interacting case $\l = \l' = 0$.
By taking the limit $L \to \io$ followed by the limit $t \to \io$, the system reaches a final steady state that is translation invariant and has a non-vanishing current that is linear in $\m_L - \m_R$.
The steady state has the same two-point correlation function as the ground state of a system of non-interacting fermions with different chemical potentials $\m_\pm = \m_0 \pm \m/2$ for right- ($+$) and left- ($-$) moving fermions, as illustrated in Fig.~\ref{Fig:Fermi_sea}.
Thus, at $t=0$ there is an excess of density on the left side compared to the right, and, asymptotically in time, there is a steady current corresponding to more right-moving fermions (coming from the left) than left-moving fermions (coming from the right).
The current satisfies the following relation in the non-interacting case:
\be
I = \frac{e^2}{h} (\m_L - \m_R) = \frac{e^2}{h} (\m_+ - \m_-),
\ee
where $e^2/h$ is the conductance quantum for the spinless case.
The conductivity defined in \eqref{df} is therefore diverging linearly with $L$ (corresponding to case~\ref{case_perfect_conductivity} in our general classification).

\begin{figure}[!htbp]
	\centering
		\begin{minipage}{.48\textwidth}
		\centering
		\includegraphics[scale=1]{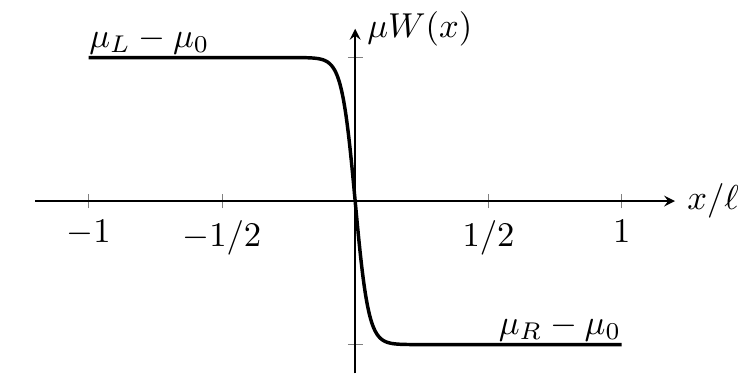}
		\caption{Domain wall initial state, with the left chemical potential $\m_L - \m_0$ larger than the right chemical potential $\m_R - \m_0$, produced by an external field $\Wint(x)$ with a smooth transition between the two sides of the system.}
		\label{Fig:External_field}
	\end{minipage}
	\hfill
	\begin{minipage}{.48\textwidth}
		\centering
		\includegraphics[scale=1]{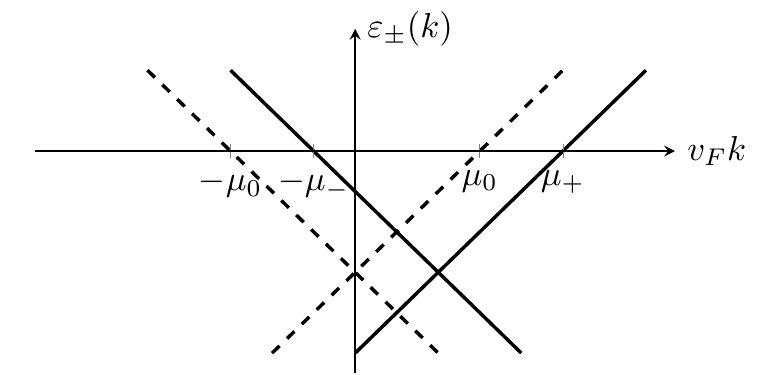}
		\caption{Fermi sea for the final state given by the linear dispersion relations $\e_\pm(k) = \pm\vF k - \m_\pm$ for right- ($+$) and left- ($-$) moving fermions starting from a domain wall state, $\m > 0$ (solid lines), and from a uniform state, $\m = 0$ (dashed lines).}
		\label{Fig:Fermi_sea}
	\end{minipage}
\end{figure}

In this paper we address the question of how the interaction modifies the above picture.
The results are presented in detail in Sec.~\ref{Sec:Results}.
As before, the system reaches a steady state that is translation invariant and has a non-vanishing current that is still linear in $\m_L - \m_R$ but depends on the Hamiltonian driving the evolution and on the initial state.
However, the final steady state obtained as $t \to \io$ has different chemical potentials $\m_+$ and $\m_-$ for right- and left-moving fermions, respectively, obtained from the two-point correlation function.
While $\m_L - \m_R = \m_+ - \m_-$ without interaction, this is not true in the interacting case:
\be
I = G_{\l,\l'} (\m_L - \m_R) = \frac{e^2}{h} (\m_+ - \m_-),
\label{I_Gll_mLMR_G_mpmm}
\ee
where $G_{\l,\l'}$ is independent of $\m_L - \m_R$ but is a non-trivial function of the microscopic parameters $\l$ and $\l'$.
The approach to this steady state is somewhat similar to the non-interacting case: we will show that the system evolves \emph{ballistically} but that the non-local interaction produces \emph{dispersion} effects, similar to what is observed in numerical simulations for quantum XXZ spin chains in \cite{SM}.
When $\l=\l'$, the ratio $G_{\l,\l'}$ between $I$ and $\m_L - \m_R$ reduces to the value computed at equilibrium in \emph{linear response theory}, $G_{\l,\l} = K_{\l} e^2/h$ \cite{KF}, where $K_{\l}$ is the so-called Luttinger parameter; this was also found by numerical simulations in \cite{SM}.

We find the relation in \eqref{I_Gll_mLMR_G_mpmm} remarkable for the following reasons.
It says that the ratio of the steady current to the difference between left and right chemical potentials is highly \emph{non-universal} (since it depends on the details of the initial state and the Hamiltonian driving the evolution).
On the other hand, the ratio of the steady current to the difference between chemical potentials of right- and left-moving fermions for the final state is perfectly \emph{universal}.
The reason is that both $I$ and $\m_+ - \m_-$ are renormalized -- a quantum many-body effect due to the interaction between the particles -- but with the property that the renormalizations precisely cancel if one takes their ratio.
This means that the Landauer conductance \cite{La2, ACF2} for the final state is universal, which confirms previous results obtained in near-to-equilibrium approaches; see, e.g., \cite{ACF1, ACF2, MS, K}.
However, we stress that we show this universality in a dynamical model of interacting fermions by a fully non-equilibrium approach.

Further understanding of the final steady state can be gained by studying the correlation functions.
When $\l=\l'$ this state has the same two-point correlation function as the ground state of a Luttinger Hamiltonian similar to $H_{\l}$ but with different chemical potentials $\m_\pm$ for right- and left-moving fermions.
This suggests that it corresponds to a generalized canonical ensemble \cite{RDYO}.
However, this is very different if $\l \neq \l' \neq 0$, where a different steady state is reached which cannot be described as the ground state of some Luttinger Hamiltonian (this can be seen by, e.g., studying non-equal-time correlation functions; cf.\ also \cite{IC}).
This steady state, which has to be a function of the constants of motion of $H_{\l'}$, has peculiar properties: its two-point correlation function has exponents which are non-trivial functions of $\l$ and $\l'$ and thus different from the equilibrium exponents.
These non-equilibrium exponents reduce to the equilibrium ones computed in \cite{ML} when $\l=\l'$ and to the ones in \cite{IC} when $\l=0$ or $\l' = 0$, but the general expressions are to our knowledge new.
Whether this state also corresponds to a generalized canonical ensemble is left open.

The paper is organized as follows.
In Sec.~\ref{Sec:Results} we formulate the model and state the results.
In Sec.~\ref{Sec:Exact_solution_of_the_Luttinger_model} we review the solution of the Luttinger model in \cite{ML}.
In Secs.~\ref{Sec:Luttinger_model_with_an_external_field}--\ref{Sec:Approach_to_steady_state} we prove our results by, first, solving the Luttinger model with an external field, second, quenching the system, and third, studying the approach to a steady state; for the latter we use a mathematical result stated and proved in Appendix~\ref{App:Asymptotics}.
In Sec.~\ref{Sec:Luttinger_model_with_constant_chemical_potentials} we prove that, for $\l = \l'$, the two-point correlation function for the final steady state is the same as that of a Luttinger model at equilibrium with different constant chemical potentials for right- and left-moving fermions; this is extended in Appendix~\ref{App:General_interactions} to more general interactions.
Sec.~\ref{Sec:Concluding_remarks} contains concluding remarks, including a discussion of the conductance in the Luttinger model.


\section{Formulation and results}
\label{Sec:Results}
We study the evolution and approach to steady state of a system of right- and left-moving interacting (spinless) fermions on a line with length $L > 0$ described by the Luttinger model with a short range non-local interaction potential $V(x)$.
The system is put out of equilibrium by switching off an external field producing an excess of density on one side compared to the other.
The time evolution is given by the Hamiltonian in \eqref{hamiltonian} with fields $\tilde\psi^\pm_{r}(x)$, which is equivalent to
\begin{multline}
H_{\l} = \sum_{r=\pm} \int_{-L/2}^{L/2} dx\,
	\! \wick{ \psi^+_{r}(x) \left(-ir \vF \partial_x \right) \psi^-_{r}(x) } \\
+ \l \sum_{r,r'=\pm} \int_{-{L/2}}^{L/2} dx\, dy\, \Vint(x-y)
	\! \wick{ \psi^+_{r\ppr}(x) \psi^-_{r\ppr}(x) } \!
	\! \wick{ \psi^+_{r'}(y) \psi^-_{r'}(y) } \!
\label{hamiltonian1}
\end{multline}
with fields $\psi^\pm_{r}(x) = L^{-1/2} \sum_{k}  a^{\pm}_{r,k} e^{\mp ikx}$ for $r = \pm$ and $k={\pi (2n+1)}/{L}$ with $n \in \mathbb{Z}$ (corresponding to antiperiodic boundary conditions) describing right- ($r = +$) and left- ($r = -$) moving fermions, where $ a^{\pm}_{r,k}$ are fermion creation and annihilation operators; this follows from the identities
\be
\tilde\psi^{\pm}_{r}(x)
= e^{\mp ir\vF^{-1} \m_0 x} \psi^{\pm}_{r}(x),
\qquad
\wick{ \tilde\psi^+_{r}(x) \tilde\psi^-_{r}(x) }
= \wick{ \psi^+_{r}(x) \psi^-_{r}(x) } + \frac{\m_0}{2\pi\vF}
\label{psi_tilde_to_psi}
\ee
(the proof is given at the end of Sec.~\ref{Sec:Luttinger_model_with_constant_chemical_potentials}).
Important (local) observables we consider are the densities $\r_\pm(x)= \wick{ \psi^+_\pm(x) \psi^-_\pm(x) }$, which satisfy periodic boundary conditions if $L$ is finite.
We define the total density as $\r(x)=\r_+(x)+\r_-(x)$ and the current as $j(x)= \vF (\r_+(x)-\r_-(x))$.
In Appendix~\ref{App:General_interactions} we show that these definitions are consistent with the continuity equation $\partial_t \r(x,t) + \partial_x j(x,t) = 0$ where $\r(x,t) = e^{iH_{\l} t} \r(x) e^{-iH_{\l} t}$ and similarly for $j(x,t)$.
We note that the second identity in \eqref{psi_tilde_to_psi} makes clear how the densities should be interpreted, namely as relative densities around the large ground state density of the filled Dirac sea determined by the Fermi momentum $\m_0/\vF$.
This explains why both positive and negative densities are allowed.

We require the following conditions on the Fourier transform $\hat{\Vint}(p) = \int_{-L/2}^{L/2} dx\, \Vint(x) e^{-ipx}$ of the interaction potential to be satisfied:
\be
\begin{aligned}
& 1. \quad \hat{\Vint}(p) = \hat{\Vint}(-p) \quad \forall p,\\
& 2. \quad \l \hat{\Vint}(p) > - \pi \vF/2 \quad \forall p,\\
& 3. \quad \hat{\Vint}(p) |p|^{1+\e} \to 0 \;\,
		\textnormal{as} \;\, p \to \pm\io \;\, \textnormal{for some} \;\, \e > 0.
\end{aligned}
\label{V_conditions_1}
\ee
We will show that the second condition (cf.\ \cite{ML}) ensures that the system is stable and that the third is needed for the interacting and non-interacting fermion Fock spaces to be unitarily equivalent.

In Sec.~\ref{Sec:Exact_solution_of_the_Luttinger_model} we make the Luttinger model mathematically precise by defining it in Fourier space (cf.\ Remark~\ref{Remark:Thermodynamic_limit}).
In particular, as is well-known, there is a unique Hilbert space, the \emph{fermion Fock space} $\mathcal{F}$,  defined by the canonical anticommutation relations and with \emph{vacuum} $|\Psi_0\rangle$ given by the ground state of $H_{0}$ (the relations determining $\mathcal{F}$ are given in \eqref{a1} and \eqref{a2}).
Moreover, the Luttinger Hamiltonian $H_{\l}$ is a well-defined self-adjoint operator on $\mathcal{F}$ bounded from below, with pure point spectrum, and with a non-degenerate ground state; see, e.g., \cite{LaMo} for a recent review of the pertinent mathematical results.

The equilibrium properties of the Luttinger model are well known and the ground state correlation functions can be exactly computed \cite{ML}.
If we let $|\Psi_{\l} \rangle$ denote the ground state of $H_{\l}$, then, in the \emph{thermodynamic limit} $L\to\io$, the two-point correlation function is
\be
\langle\Psi_{\l}| \psi^+_{r}(x) \psi^-_{r}(y) |\Psi_{\l}\rangle
= \frac{i}{2\pi r(x-y)+i0^+}
	\exp \left( \int_0^\io dp \frac{\h_{\l}(p)}{p} (\cos p(x-y) - 1)\right)
\label{ss}
\ee
with $\h_{\l}(p) = (1 - [\l\hat{\Vint}(p)/(\pi\vF + \l\hat{\Vint}(p))]^2)^{-1/2} - 1$.
As emphasized, this result (derived in \cite{ML} for a different interaction) is in the thermodynamic limit; this is motivated by our interest in length scales much smaller than the system size as discussed in Sec.~\ref{Sec:Introduction}.
For later reference we also note that \eqref{psi_tilde_to_psi} implies
\be
\langle\Psi_{\l}| \tilde\psi^+_{r}(x) \tilde\psi^-_{r}(y) |\Psi_{\l}\rangle
= e^{-ir\vF^{-1} \m_0 (x-y)}
	\langle\Psi_{\l}| \psi^+_{r}(x) \psi^-_{r}(y) |\Psi_{\l}\rangle.
\label{ss_tilde}
\ee
The interaction produces a dramatic modification in the long distance decay of the correlation functions, with the appearance of an \emph{anomalous} exponent $\h_{\l} = \h_{\l}(0)$, where, for large $|x-y|$, the two-point correlation function decays as $O(|x-y|^{-1-\h_{\l}})$.

We investigate the non-equilibrium properties of this system when the initial state has an excess of density on one side compared to the other.
We choose as initial state the ground state of the Hamiltonian in \eqref{hamiltonian_with_muL_and_muR}, which, using \eqref{psi_tilde_to_psi}, is equivalent to
\be
H_{\l,\m} = H_{\l} - \m \int_{-L/2}^{L/2} dx\, \Wint(x)\r(x),
\label{hamiltonian_with_ext_field}
\ee
where $\Wint(x)$ is an external field taken as a regularized version of $1/2-\th(x)$ (with the Heaviside function $\th(x) = 1$ for $x > 0$ and $\th(x) = 0$ for $x < 0$) as depicted in Fig.~\ref{Fig:External_field}.
For finite $L$, we need to take the periodic boundary conditions into account, and for this we use
\be
\Wint(x) = - \frac{1}{2} \left( \tanh (x/\d) - \tanh((2x+L)/2\d) - \tanh((2x-L)/2\d) \right)
\label{WxL}
\ee
for some small $\d > 0$.
In Sec.~\ref{Sec:Luttinger_model_with_an_external_field} we make this model well-defined by again working in Fourier space.
By explicitly constructing all eigenstates and the corresponding eigenvalues of $H_{\l,\m}$ we show the following:

\begin{theorem}
\label{Thm:H_lambda_mu}
For finite $L$, the Hamiltonian $H_{\l,\m}$ in \eqref{hamiltonian_with_ext_field}, with $\Vint(x)$ satisfying the conditions in \eqref{V_conditions_1} and $\Wint(x)$ in \eqref{WxL}, defines a self-adjoint operator on the fermion Fock space $\mathcal{F}$.
This operator is bounded from below, has pure point spectrum, and a non-degenerate ground state.
\end{theorem}

From Sec.~\ref{Sec:Introduction} we recall that the ground state of $H_{\l,\m}$ is denoted by $|\Psi_{\l,\m} \rangle$, and we consider the evolution of this state under a different Hamiltonian $H_{\l'} = H_{\l',0}$ with coupling constant $\l'$ and no external field, i.e., we quench the system and consider
\be
|\Psi_{\l,\m}^{\l'}(t)\rangle = e^{-iH_{\l'}t}|\Psi_{\l,\m}\rangle.
\label{Psi_lp_l_m}
\ee
We show that this state tends to a steady state as $t \to \io$ by analytically computing the expectation values of certain observables.
It is at the level of these expectation values that we pass to the thermodynamic limit (cf.\ Remark~\ref{Remark:Thermodynamic_limit}), and in this limit we use the regularized external field $\Wint(x) = -(1/2) \tanh(x/\d)$.
This external field has the Fourier transform
\be
\hat{\Wint}(p) = \frac{i\pi\d}{2\sinh(\pi\d p/2)},
\label{reg_Wint_Fourier}
\ee
and letting $\d \to 0^+$ yields $\hat{\Wint}(p) = ip^{-1}$ which is the Fourier transform of $1/2-\th(x)$ (interpreted in a distributional sense as a Cauchy principal value).
We also impose the following conditions on the interaction potential in order to derive exact results for the asymptotical behavior of the system: 
\be
\begin{aligned}
& 1. \quad \hat{\Vint}(p) \in C^{2}(\mathbb{R}) \;\, \textnormal{(a.e.)},\\
& 2. \quad \hat{\Vint}(p),\;
		d\hat{\Vint}(p)/dp,\;
		d^2\hat{\Vint}(p)/dp^2 \in L^{1}(\mathbb{R}),\\
& 3. \quad \l p\, d\hat{\Vint}(p)/dp > -\pi\vF - 2\l\hat{\Vint}(p) \quad \forall p.
\end{aligned}
\label{V_conditions_2}
\ee
As will be shown, the third condition means that the system evolves with a positive group velocity (cf.\ \eqref{vG}).

\begin{oss}
\label{Remark:Thermodynamic_limit}
Following the approach described in Sec.~\ref{Sec:Introduction}, we only define the Luttinger model for finite $L$.
This allows for a simple rigorous construction of the model in Fourier space; see, e.g., \cite{LaMo}.
It is only for expectation values of observables, after we have computed them for finite $L$, that we pass to the thermodynamic limit $L \to \io$.
To make this clear, we use $\langle \cdot \rangle_{L}$ to denote expectation values for finite $L$ and write $\langle \cdot \rangle = \lim_{L\to\io} \langle \cdot \rangle_{L}$, and similarly for other quantities.
\end{oss}

We first consider the case without interaction between the fermions, i.e., $\l=\l'=0$.
The results follow as special cases from the proofs in Sec.~\ref{Sec:Evolution_following_a_quench}.
For the total density and the current we show that
\begin{align}
\langle\Psi_{0,\m}^{0}(t)| \r(x ) |\Psi_{0,\m}^{0}(t)\rangle
	& = \frac{\m}{2\pi\vF} \left( \Wint(x - \vF t) + \Wint(x + \vF t) \right),
	\label{Rxt0} \\
\langle\Psi_{0,\m}^{0}(t)|j(x)|\Psi_{0,\m}^{0}(t)\rangle
	& = \frac{\m}{2\pi} \left( \Wint(x - \vF t) - \Wint(x + \vF t) \right).
	\label{Ixt0}
\end{align}
If $\Wint(x) = {1/2} - \th(x)$, or our regularized version thereof, this means that there is a central region $(- \vF t, \vF t)$ around $x=0$ with zero total density (relative to the large constant ground state density) bounded by two fronts moving with constant velocity.
The shape of the fronts does not change with time, and, as $t\to\io$, the system reaches a state with vanishing total density everywhere.
Similarly, the current is non-zero in the same region, and, as $t\to\io$, it tends to the non-vanishing value $\m/2\pi$ everywhere.

We also show that the two-point correlation function without interaction is given by
\be
\langle\Psi_{0,\m}^{0}(t)| \psi^+_{r}(x)\psi^-_{r}(y)
	|\Psi_{0,\m}^{0}(t)\rangle
= \frac{i}{2\pi r(x-y)+i 0^+}
	\exp \left( -ir\vF^{-1} \m \int_{y-r\vF t}^{x-r\vF t} dz \Wint(z) \right).
\ee
For finite $t$, the two-point correlation function is not translation invariant.
However, asymptotically in time,
\be
\lim_{t \to \io} \int_{y-r\vF t}^{x-r\vF t} dz \Wint(z) = r \frac{x-y}{2},
\ee
meaning that translation invariance is recovered:
\be
\lim_{t\to\io}
\langle\Psi_{0,\m}^{0}(t)| \psi^+_{r}(x) \psi^-_{r}(y) |\Psi_{0,\m}^{0}(t)\rangle
= \frac{ie^{-ir\vF^{-1} (r\m/2) (x-y)}}{2\pi r(x-y)+i 0^+}.
\label{acc0}
\ee
Since \eqref{acc0} is similar to \eqref{ss} with $\l=0$, this suggests that the final steady state is similar to the ground state of free fermions with different chemical potentials $\m_\pm - \m_0 = \pm \m/2$ for right- and left-moving fermions, obtained from the two-point correlation function (cf.\ \eqref{ss_tilde}).
We can make this precise by comparing \eqref{acc0} with the two-point correlation function obtained from a Hamiltonian describing such a system; this we will do below for the case with interaction.
Therefore, in absence of interaction, a very clear picture emerges: the ground state of free fermions, with an external field producing a domain wall, evolves as $t\to\io$ to a steady state which has the same two-point correlation function as that of a ground state of free fermions with different chemical potentials for right- and left-moving fermions.

Let us now consider the evolution of the domain wall initial state in the presence of interaction, i.e., with non-zero $\l$ and $\l'$.
The question we try to answer is the following: how does the interaction modify the evolution of the domain wall state?
Below we present the exact results for general $\l$ and $\l'$.
The proofs are given in Sec.~\ref{Sec:Evolution_following_a_quench}.

We show that the total density and the current are given by the following exact expressions in the thermodynamic limit:
\begin{align}
R(x,t) = \langle\Psi_{\l,\m}^{\l'}(t)| \r(x) |\Psi_{\l,\m}^{\l'}(t)\rangle
	& = \frac{\m}{2\pi} \int_{-\io}^{\io} \frac{dp}{2\pi}
		\frac{K_{\l}(p)}{v_{\l}(p)} \hat{\Wint}(p) 
		2\cos(pv_{\l'}(p)t) e^{ipx},
\label{Rxt} \\
I(x,t) = \langle\Psi_{\l,\m}^{\l'}(t)| j(x) |\Psi_{\l,\m}^{\l'}(t)\rangle
	& = \frac{\m}{2\pi} \int_{-\io}^{\io} \frac{dp}{2\pi}
		\frac{K_{\l}(p)}{v_{\l}(p)} \hat{\Wint}(p) v_{\l'}(p)
		(-2i \sin(pv_{\l'}(p)t)) e^{ipx}
\label{Ixt}	
\end{align}
with the \emph{renormalized Fermi velocity}
\be
v_{\l}(p) = \vF \sqrt{ 1 + 2\l\hat{\Vint}(p)/\pi\vF}
\label{vL}
\ee
and the \emph{Luttinger parameter}
\be
K_{\l}(p) = \frac{1}{\sqrt{1 + 2\l\hat{\Vint}(p)/\pi\vF}}.
\label{Luttinger_parameter}
\ee
For our particular interaction $K_{\l}(p)$ and $v_{\l}(p)$ satisfy $K_{\l}(p) v_{\l}(p) = \vF$, but we note that this is not true in general; see, e.g., \cite{V} or Appendix~\ref{App:General_interactions}.
For a special case, similar to a case considered in \cite{IC}, we plot the total density in Fig.~\ref{Fig:Total_density} and the current in Fig.~\ref{Fig:Current} for the subsystem on the interval $[-\ell, \ell]$ with $L > \ell > 0$.
In this case, the initial state is the non-interacting ground state $|\Psi_{0,\m}\rangle$, which is evolved under $H_{\l'}$ with the non-local interaction potential $\hat{\Vint}(p) = (\pi\vF/2) \sech(ap)$ with interaction range $a > 0$.

We show that the two-point correlation function is given by the following exact expression in the thermodynamic limit:
\be
\langle\Psi_{\l,\m}^{\l'}(t)| \psi^+_{r}(x)\psi^-_{r}(y) |\Psi_{\l,\m}^{\l'}(t)\rangle
= e^{-ir\vF^{-1} A_{r}(x,y,t)(x-y)}
	S_{r}(x,y,t)
\label{psi_psi_factoring_infinite_L}
\ee
with
\be
A_{r}(x,y,t)
= \m \int_{-\io}^{\io} \frac{dp}{2\pi}
		\frac{K_{\l}(p)}{v_{\l}(p)} \hat{\Wint}(p)
		\left(
			\vF \cos(pv_{\l'}(p)t) - ir v_{\l'}(p) \sin(pv_{\l'}(p)t)
		\right) \frac{e^{ipx} - e^{ipy}}{ip(x-y)}
\label{Arxyt}
\ee
and
\be
S_{r}(x,y,t)
= \langle\Psi_{\l,0}^{\l'}(t)| \psi^+_{r}(x)\psi^-_{r}(y) |\Psi_{\l,0}^{\l'}(t)\rangle.
\label{Srxyt}
\ee
The latter is the two-point correlation function in the absence of external field, i.e., $\m = 0$, in which case the initial state is the ground state $|\Psi_{\l}\rangle$ of $H_{\l}$ but the Hamiltonian driving the evolution is $H_{\l'}$ as before.
We show that
\be
S_{r}(x,y,t)
= \frac{i}{2\pi r(x-y) + i0^+}
\exp \left( \int_0^\io dp
	\frac{\h_{\l,\l'}(p) - \g_{\l,\l'}(p) \cos (2p v_{\l'}(p)t)}{p}
	\left( \cos p(x-y) - 1 \right)
\right)
\label{l_lp_two_point_func}
\ee
with exponents
\be
\begin{aligned}
\h_{\l,\l'}(p)
	& = \frac{ K_{\l}(p) (K_{\l'}(p)^{-2}+1) + K_{\l}(p)^{-1}(K_{\l'}(p)^{2}+1) }{4} - 1, \\
\g_{\l,\l'}(p)
	& = \frac{ K_{\l}(p) (K_{\l'}(p)^{-2}-1) + K_{\l}(p)^{-1}(K_{\l'}(p)^{2}-1) }{4}.
\end{aligned}
\label{eta_gamma_K_Kp}
\ee
In general, these differ from the equilibrium exponents in \eqref{ss}.
The latter are obtained only when $\l = \l'$, in which case $\h_{\l,\l}(p) = \h_{\l}(p) = (K_{\l}(p) + K_{\l}(p)^{-1})/2 - 1$ and $\g_{\l,\l}(p) = 0$, where $\h_{\l}(p)$ can be expressed in terms of $\l \hat{\Vint}(p)$ using \eqref{Luttinger_parameter}.
Note also the identity
\be
\h_{\l,\l'}(p) = \h_{\l}(p) + \g_{\l,\l'}(p).
\label{eta_l_lp_equal_eta_l_gamma_l_lp}
\ee
This shows that the correlation function in \eqref{ss} is, indeed, recovered at time $t = 0$.

\begin{figure}[!htbp]
	\centering
	
	\subfigure[\;$t=0.0\ell/\vF$ and $\l' = -0.96$]{
		\includegraphics[scale=1]{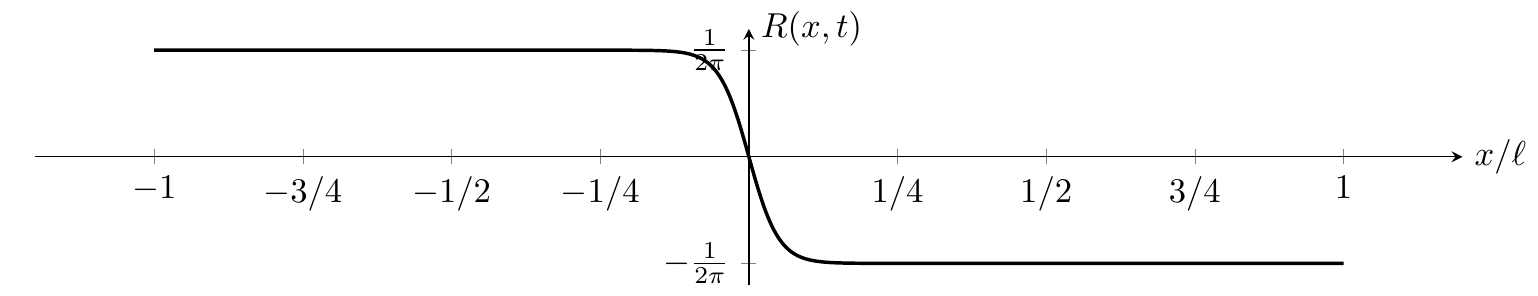}
		\label{Fig:Density_t_0}
	}
	
	\subfigure[\;$t=2.0\ell/\vF$ and $\l' = -0.96$]{
		\includegraphics[scale=1]{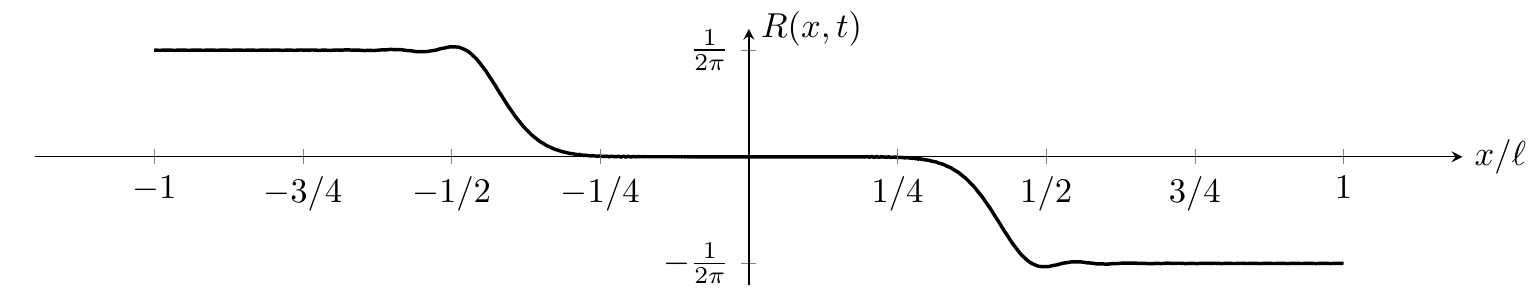}
		\label{Fig:Density_t_400}
	}
	
	\subfigure[\;$t=4.0\ell/\vF$ and $\l' = -0.96$]{
		\includegraphics[scale=1]{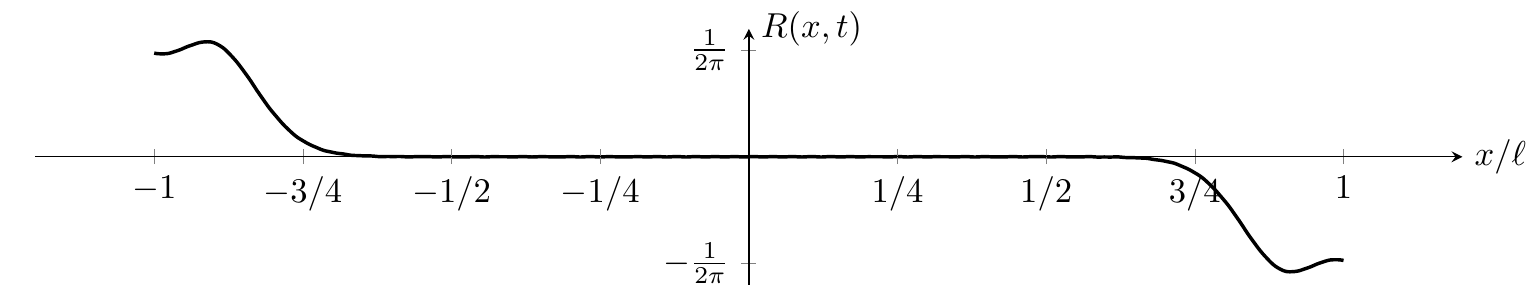}
		\label{Fig:Density_t_800}
	}
	
	\subfigure[\;$t=6.0\ell/\vF$ and $\l' = -0.96$]{
		\includegraphics[scale=1]{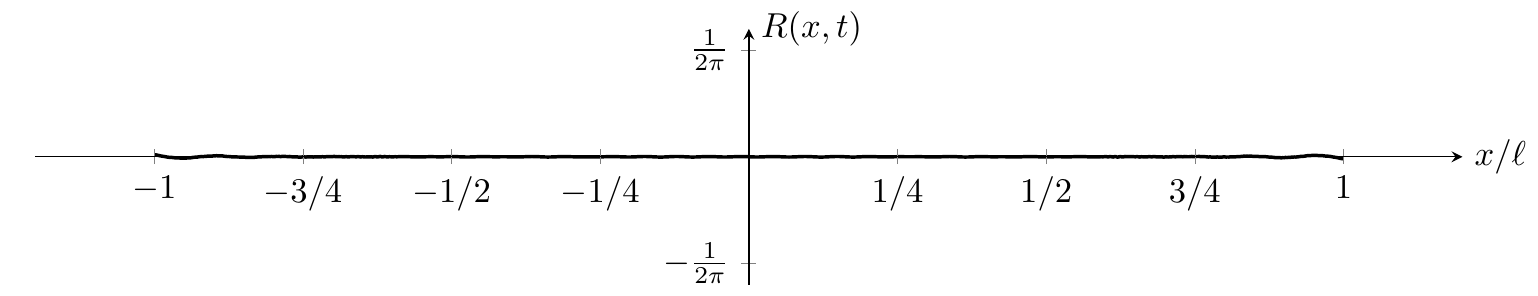}
		\label{Fig:Density_t_1200}
	}
	
	\subfigure[\;$t = 2.0\ell/\vF$ and $\l' = -0.80$ (solid line), $\l' = -0.90$, $\l' = -0.95$, and $\l' = -0.99$ (dotted line)]{
		\includegraphics[scale=1]{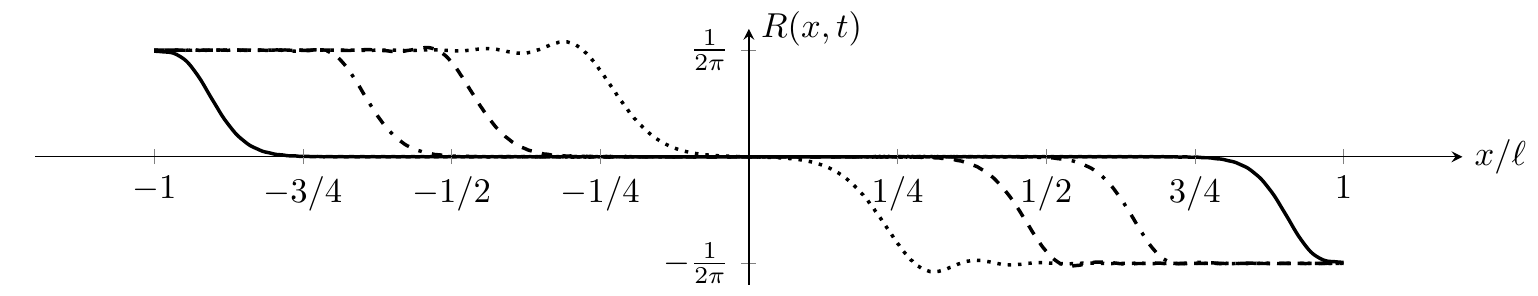}
		\label{Fig:Density_lambda_all}
	}
	
	\caption{Total density at time $t$ and for coupling constant $\l'$, starting from the non-interacting ground state, $\l = 0$, with domain wall profile, $\m = 1$, and evolving with a non-local interaction potential $\hat{\Vint}(p) = (\pi\vF/2) \sech(ap)$ with interaction range $a = 0.0025\ell$ and $\vF = 1$.}
	\label{Fig:Total_density}
\end{figure}
\begin{figure}[!htbp]
	\centering
	
	\subfigure[\;$t=0.0\ell/\vF$ and $\l' = -0.96$]{
		\includegraphics[scale=1]{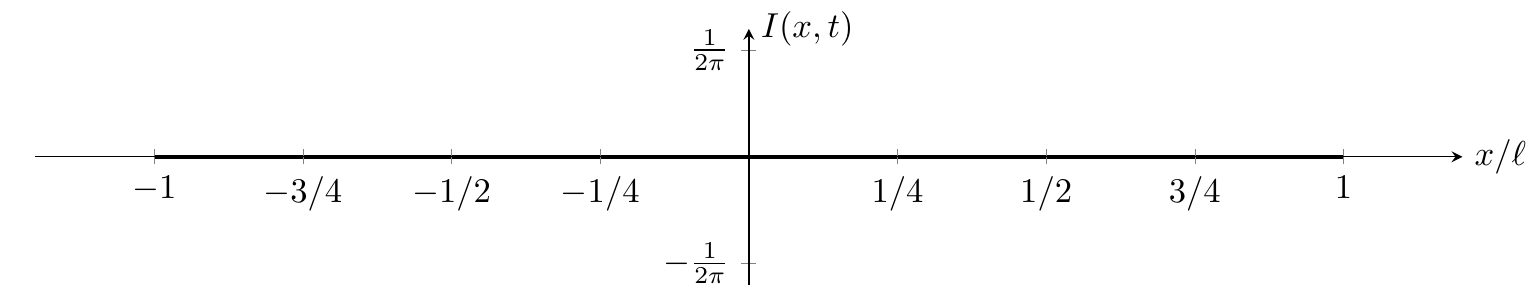}
		\label{Fig:Current_t_0}
	}
	
	\subfigure[\;$t=2.0\ell/\vF$ and $\l' = -0.96$]{
		\includegraphics[scale=1]{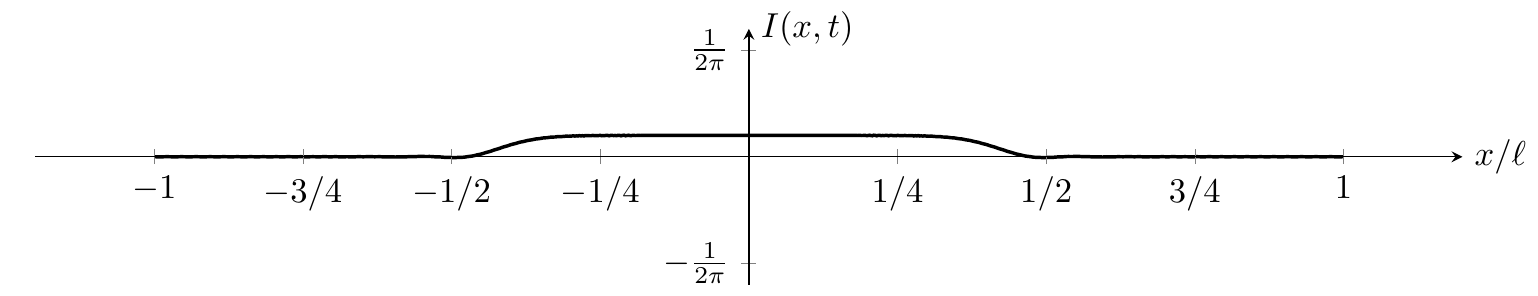}
		\label{Fig:Current_t_400}
	}
	
	\subfigure[\;$t=4.0\ell/\vF$ and $\l' = -0.96$]{
		\includegraphics[scale=1]{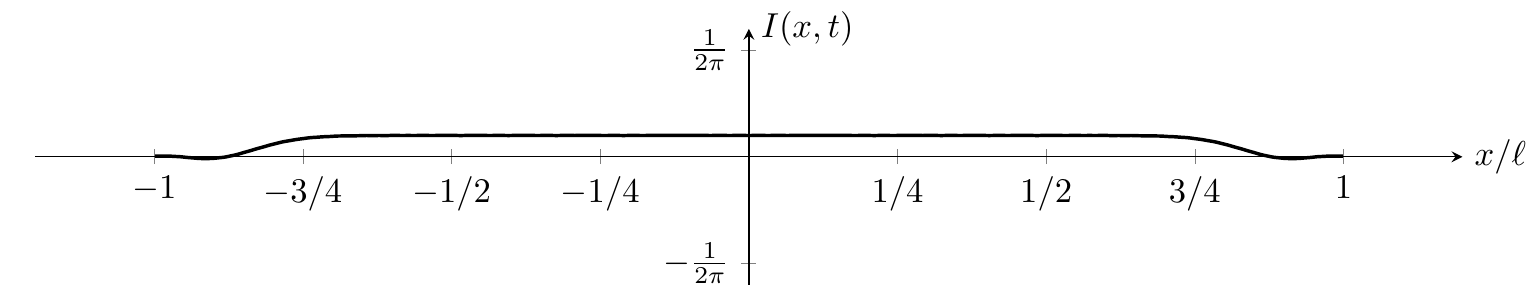}
		\label{Fig:Current_t_800}
	}
	
	\subfigure[\;$t=6.0\ell/\vF$ and $\l' = -0.96$]{
		\includegraphics[scale=1]{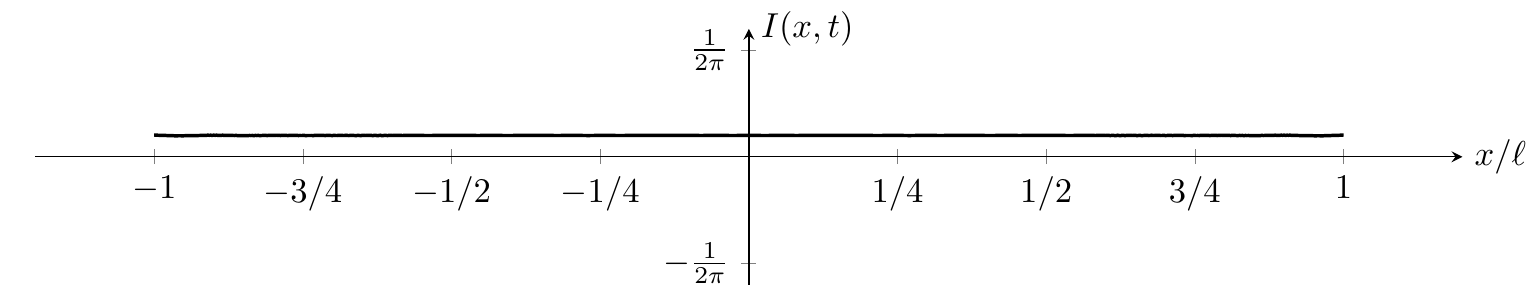}
		\label{Fig:Current_t_1200}
	}
	
	\subfigure[\;$t = 2.0\ell/\vF$ and $\l' = -0.80$ (solid line), $\l' = -0.90$, $\l' = -0.95$, and $\l' = -0.99$ (dotted line)]{
		\includegraphics[scale=1]{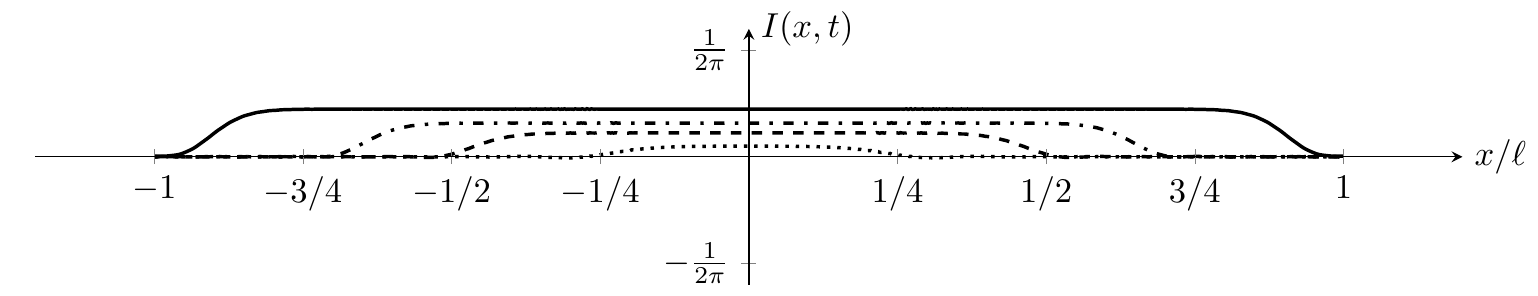}
		\label{Fig:Current_lambda_all}
	}
	
	\caption{Current at time $t$ and for coupling constant $\l'$, starting from the non-interacting ground state, $\l = 0$, with domain wall profile, $\m = 1$, and evolving with a non-local interaction potential $\hat{\Vint}(p) = (\pi\vF/2) \sech(ap)$ with interaction range $a = 0.0025\ell$ and $\vF = 1$.}
	\label{Fig:Current}
\end{figure}

Using the results above we study the asymptotic behavior of the total density, the current, and the two-point correlation function.
One important quantity describing the evolution is the group velocity
\be
v^{\textnormal{g}}_{\l}(p) = d(pv_{\l}(p))/dp,
\label{vG}
\ee
which we require to be positive (cf.\ Lemma~\ref{Lemma:Asymptotics} in Appendix~\ref{App:Asymptotics} where the positivity of $v^{\textnormal{g}}_{\l}(p)$ is needed).
This can be shown to correspond to the third condition in \eqref{V_conditions_2}.
In Sec.~\ref{Sec:Approach_to_steady_state} we prove the following result (see \eqref{R_I_A_asymptotics} for a stronger version of \eqref{Rxt_Ixt_Arxyt_Asymptotics}):

\begin{theorem}
\label{Thm:Asymptotics}
If $\Vint(x)$ satisfies the conditions in \eqref{V_conditions_1} and \eqref{V_conditions_2}, then
\be
\lim_{t \to \io} R(x,t) = 0,
\qquad
\lim_{t \to \io} I(x,t) = \frac{\m}{2\pi} \frac{K_{\l} v_{\l'}}{v_{\l}},
\qquad
\lim_{t \to \io} A_{\pm}(x,y,t) = \pm \frac{\m}{2} \frac{K_{\l} v_{\l'}}{v_{\l}}
\quad
\forall x,y
\label{Rxt_Ixt_Arxyt_Asymptotics}
\ee
and
\begin{multline}
\lim_{t\to\io}
\langle\Psi_{\l,\m}^{\l'}(t)| \psi^+_{r}(x) \psi^-_{r}(y) |\Psi_{\l,\m}^{\l'}(t)\rangle \\
= \frac{ie^{- ir\vF^{-1} (r\m K_{\l} v_{\l'}/2v_{\l}) (x-y)}}
		{2\pi r(x-y) + i0^+}
	\exp \left( \int_0^\io dp \frac{\h_{\l,\l'}(p)}{p}
		\left( \cos p(x-y) - 1 \right) \right),
\label{acc}
\end{multline}
where $K_{\l} = K_{\l}(0)$ and $v_{\l} = v_{\l}(0)$.
\end{theorem}

One consequence of \eqref{acc} is that, while the two-point correlation function in \eqref{l_lp_two_point_func} is not translation invariant for finite $t$, translation invariance is recovered asymptotically in time.
This is a generalization of the corresponding result derived in \cite{MW}.
As in the non-interacting case, by comparing with \eqref{ss}, it follows that \eqref{acc} describes fermions with different chemical potentials for right- and left-moving particles,
\be
\m_\pm - \m_0 = \pm \frac{\m}{2} \frac{K_{\l} v_{\l'}}{v_{\l}},
\label{mr_m0_equal_Ar}
\ee
obtained from the phase factor in the correlation function (cf.\ \eqref{ss_tilde}).
However, since $\h_{\l,\l'}(p) \neq \h_{\l}(p)$ when $\l \neq \l' \neq 0$, the correlation function in \eqref{acc} is not, in general, a ground state two-point correlation function as in \eqref{ss}.
The only case where this could be true is if $\l = \l'$, and in Sec.~\ref{Sec:Luttinger_model_with_constant_chemical_potentials} we prove the following result allowing us to compare \eqref{acc} with the corresponding (equal-time) two-point correlation function:

\begin{theorem}
\label{Thm:Hamiltonian_2}
For finite $L$, the Hamiltonian
\be
H = H_{\l} - \sum_{r=\pm} (\m_{r} - \m_0) Q_{r},
\qquad
Q_{r} = \int_{-L/2}^{L/2} dx\, \r_r(x),
\label{hamiltonian2}
\ee
with $\Vint(x)$ satisfying the conditions in \eqref{V_conditions_1} and $\m_\pm$ being constant chemical potentials satisfying $\m_+ + \m_- = 2\m_0$, defines a self-adjoint operator on the fermion Fock space $\mathcal{F}$.
This operator is bounded from below, has pure point spectrum, and a non-degenerate ground state.
Denoting the ground state of $H$ by $|\Psi\rangle$, the two-point correlation function in the thermodynamic limit is
\be
\langle\Psi| \psi^+_{r}(x) \psi^-_{r}(y) |\Psi\rangle
= \frac{ie^{-ir\vF^{-1} (\m_{r}-\m_0) (x-y)}}{2\pi r(x-y)+i0^+}
	\exp \left( \int_0^\io dp \frac{\h_{\l}(p)}{p} (\cos p(x-y) - 1)\right)
\label{psi_psi_correl_2}
\ee
with $\h_{\l}(p) = \h_{\l,\l}(p)$ in \eqref{eta_gamma_K_Kp}.
\end{theorem}

Thus, \eqref{acc} for $\l = \l'$ and \eqref{psi_psi_correl_2} are the same, and this justifies the identification in \eqref{mr_m0_equal_Ar}.
Moreover, while we do not give the details, this result can also be extended to non-equal-time two-point correlation functions, as explained at the end of Sec.~\ref{Sec:Approach_to_steady_state}.

In conclusion, as $t\to\io$, the system reaches a translation invariant steady state with a vanishing density and a non-vanishing current.
The region of vanishing density increases with time and lies between two fronts which evolve ballistically.
However, for a non-local interaction, as in Figs.~\ref{Fig:Total_density}~and~\ref{Fig:Current}, the system is \emph{dispersive}.
To see this, note that, contrary to the non-interacting case, the renormalized Fermi velocity in \eqref{vL} and therefore also the group velocity in \eqref{vG}, which determines the propagation of the fronts, depend on the momentum.
This means that the shape of the fronts changes over time.
These dispersion effects appear, e.g., as oscillations traveling ahead of the fronts in Figs.~\ref{Fig:Total_density}~and~\ref{Fig:Current}.
Similar behavior was found in numerical simulations for quantum XXZ spin chains in \cite{SM}.

Only in the case $\l=\l'$, i.e., if the quench consists simply of switching off the external field, does one find that the steady state has the same two-point correlation function as the ground state of the Hamiltonian in \eqref{hamiltonian2} with different chemical potentials $\m_\pm$ for right- and left-moving fermions.
When $\l \neq \l' \neq 0$ the final steady state is inherently different from a ground state of such a Hamiltonian since the non-equilibrium exponents in \eqref{acc} are different from the equilibrium ones in \eqref{ss}.
This implies that memory of the initial state is conserved for infinite times; cf.\ also \cite{IC}.

It follows from Theorems~\ref{Thm:Asymptotics}~and~\ref{Thm:Hamiltonian_2} that the final steady state carries a current $I = \m K_{\l} v_{\l'}/2\pi v_{\l}$ and has an effective chemical potential difference $\m_+ - \m_- = \m K_{\l} v_{\l'}/v_{\l}$ depending on the coupling constants $\l$ and $\l'$.
However, even though the final state depends on the details of the time evolution and the initial state, the Landauer conductance \cite{La2, ACF2} is universal:
\be
G = \frac{I}{\m_+ - \m_-}
	= \frac{\m K_{\l} v_{\l'}}{2\pi v_{\l}} \frac{v_{\l}}{\m K_{\l} v_{\l'}}
	= \frac{1}{2\pi},
\label{Landauer_conductance_cancelation}
\ee
which is equal to the conductance quantum $e^2/h = 1/2\pi$ (in units where $e = \hbar = 1$) for the spinless case.
Since, in general, the final steady state is not the ground state of a Luttinger Hamiltonian, this universality is a true non-equilibrium phenomenon.

Finally, we note that the formulas given above for the evolution following a quench can be shown to remain valid for more general interactions in the Luttinger model; cf.\ the final remark in Appendix~\ref{App:General_interactions}.


\section{Exact solution of the Luttinger model}
\label{Sec:Exact_solution_of_the_Luttinger_model}
We first solve the model in the absence of external field.
What follows essentially reviews the solution of the Luttinger model in \cite{ML} but for a different interaction (cf.\ the appendix in \cite{ML}) than the usual Luttinger interaction \cite{Lu} (cf.\ also Appendix~\ref{App:General_interactions}).
Note that here and in the reminder of the paper we use the convention that sums over variables range over all allowed values, unless specified otherwise.

The ground state of the non-interacting Hamiltonian $H_{0}$ is the filled Dirac sea $|\Psi_0\rangle$ defined by the condition
\be
a^-_{r,rk} |\Psi_0\rangle = a^+_{r,-rk} |\Psi_0\rangle = 0 \quad \forall k > 0
\label{a1}
\ee
using fermion creation and annihilation operators $a^\pm_{r,k}$, $r = \pm$ and $k = \pi (2n+1)/L$ with $n \in \mathbb{Z}$, satisfying
\be
\left\{ a^-_{r,k\ppr}, a^+_{r',k'} \right\} = \d_{r,r'} \d_{k,k'},
\qquad
\left\{ a^\pm_{r,k\ppr}, a^\pm_{r',k'} \right\} = 0.
\label{a2}
\ee
The conditions in \eqref{a1} and \eqref{a2} fully determine the fermion Fock space $\mathcal{F}$, and the $a^\pm_{r,k}$ can be used to construct fields $\psi^\pm_{r}(x) = L^{-1/2} \sum_{k} a^\pm_{r,k} e^{\mp ikx}$ describing right- ($r = +$) and left- ($r = -$) moving fermions.
We define the densities as $\r_{r}(p) = \sum_{k} \wick{ a^+_{r,k}a^-_{r,k+p} }$ for $p={2\pi n}/{L}$ with $n \in \mathbb{Z}$, where Wick (normal) ordering $\wick{\cdots}$ can be defined as
\be
\wick{A} = A - \langle\Psi_0| A |\Psi_0\rangle
\label{wick_ordering}
\ee
for operators $A = a^+_{r,k\ppr}a^-_{r',k'}$.
The densities satisfy
\be
\begin{gathered}
\left[ \r_{r\ppr}(p),\r_{r'}(-p') \right]
	= r\delta_{r,r'} \frac{Lp}{2\pi} \delta_{p,p'}, \\
\r_+(p)|\Psi_0\rangle = \r_-(-p)|\Psi_0\rangle = 0 \quad \forall p \geq 0,
\label{fg}
\end{gathered}
\ee
and expressing the non-interacting Hamiltonian in Fourier space in its bosonized form,
\be
H_{0} = \frac{\pi}{L} \vF  \left( \r_+(0)^2 + \r_-(0)^2 \right)
	+ \sum_{p>0} \frac{2\pi}{L} \vF \left( \r_+(-p)\r_{+}(p) + \r_-(p)\r_-(-p) \right),
\ee
one can see that
\be
[H_{0}, \r_{r}(p)] = - r\vF p \r_{r}(p).
\label{H0_rsp}
\ee
These results are well-known in both the condensed matter and mathematical physics literature; see, e.g., \cite{LaMo} for self-contained proofs.
(We note that our conventions differ from \cite{ML} in that $\r_1(-p)$ and $\r_2(-p)$ in \cite{ML} correspond to our densities $\r_+(p)$ and $\r_-(p)$, respectively, and our conventions for the Fourier transform are such that $\r_{r}(x) = \sum_{p} (1/L) \r_{r}(p) e^{ipx}$ for the densities, while in \cite{ML} the corresponding transform has a minus sign in the exponent.)

We now consider the interacting Hamiltonian
\be
\begin{aligned}
H_{\l} & = H_{0}  + \l\Hint, \\
\Hint & = \frac{1}{L} \hat{\Vint}(0) \r(0)^2
	+ \sum_{p \neq 0} \frac{1}{L} \hat{\Vint}(p) \r(p) \r(-p)
\end{aligned}
\ee
with $\r(p) = \r_+(p) + \r_-(p)$ and $\hat{\Vint}(p)$ satisfying the conditions in \eqref{V_conditions_1}.
As explained in \cite{ML}, $H_{\l}$ can be diagonalized by a Bogoliubov transformation implemented by a unitary operator $e^{iS_{\l}}$ (defined in \eqref{S_operator}):
\be
e^{i S_{\l}} H_{\l} e^{-i S_{\l}} = H_{0} - T_{\l} + D_{\l} + W_{\l}
\label{Diagonalized_Hamiltonain}
\ee
with
\be
\begin{aligned}
T_{\l} & = - \frac{1}{L} \l \hat{\Vint}(0) \r(0)^2
	+ \sum_{p>0} \frac{2\pi}{L} \vF \left( \r_+(-p)\r_{+}(p) + \r_-(p)\r_-(-p) \right), \\
D_{\l} & = \sum_{p>0} \frac{2\pi}{L} v_{\l}(p)
		 \left( \r_+(-p)\r_{+}(p) + \r_-(p)\r_-(-p) \right), \\
W_{\l} & = - \sum_{p>0} \left( \vF - v_{\l}(p) \right) p, \vphantom{\frac{2\pi}{L}}
\end{aligned}
\label{Diagonalized_Hamiltonain_ops}
\ee
where $v_{\l}(p)$ is the renormalized Fermi velocity in \eqref{vL}.
Since $H_{0} - T_{\l}$ contains only zero modes and $W_{\l}$ is a c-number, both commute with the densities $\r_{r}(p)$.
The unitary operator $e^{iS_{\l}}$ is given by
\be
S_{\l} = i \sum_{p \neq 0} \frac{2\pi}{L} \frac{\ph_{\l}(p)}{p} \r_{+}(-p)\r_-(p),
\qquad
\tanh 2\ph_{\l}(p) = -\frac{\l\hat{\Vint}(p)}{\pi\vF+\l\hat{\Vint}(p)}
\quad
\forall p \neq 0.
\label{S_operator}
\ee
It follows that $\left[ S_{\l}, \r_{r}(p) \right] = -i\ph_{\l}(p) \r_{-r}(p)$ for $p \neq 0$, which means that
\be
e^{i S_{\l}} \r_{r}(p) e^{-i S_{\l}}
	= \r_{r}(p)\cosh\ph_{\l}(p) + \r_{-r}(p)\sinh\ph_{\l}(p)
\label{epiS_rsp_emiS}
\ee
(to see this, define $g_{r}(u) = e^{iu S_{\l}} \r_{r}(p) e^{-iu S_{\l}}$, differentiate, and solve the resulting differential equation $g'_{r}(u) = \ph_{\l}(p) g_{-r}(u)$ with $g_{r}(0) = \r_{r}(p)$).

The interacting ground state is $|\Psi_\l\rangle = e^{-iS_{\l}} |\Psi_0\rangle$ with ground state energy $W_{\l}$.
Since the latter has to be a finite number, this imposes the condition in \eqref{V_conditions_1} that $\hat{\Vint}(p)$ must decay faster than $|p|^{-1}$ for large $|p|$.
We also note that, when deriving \eqref{Diagonalized_Hamiltonain_ops}, $v_{\l}(p)$ is found as
\be
v_{\l}(p) = \vF e^{-2\ph_{\l}(p)},
\label{vL_ph}
\ee
which, using \eqref{S_operator}, gives precisely \eqref{vL}.
(The relation in \eqref{vL_ph} does \emph{not} hold for the general interactions considered in Appendix~\ref{App:General_interactions}; cf.\ the comment below \eqref{Luttinger_parameter}.)
From this we see that, since $v_{\l}(p)$ has to be real, we must impose the stability condition $\l \hat{\Vint}(p) > - \pi \vF/2$ in \eqref{V_conditions_1}.
Moreover, similar to \eqref{H0_rsp},
\be
[D_{\l}, \r_{r}(p)] = - r v_{\l}(p) p \r_{r}(p),
\label{D_rsp}
\ee
which makes clear that $v_{\l}(p)$ is, indeed, to be interpreted as the renormalized Fermi velocity.

Finally, we demonstrate how to construct the eigenstates of $H_{\l}$ from the eigenstates of $H_{0}$.
To construct the latter we need the operators $Q_{r} = \r_{r}(p=0)$ in \eqref{hamiltonian2}, which are Hermitian, and also to introduce new operators $R_{r}$ called \emph{Klein factors}, which are unitary, satisfying
\be
[Q_{r}, R_{r'}] = r\d_{r,r'}R_{r},
\qquad
R_\pm R_\mp = -R_\mp R_\pm,
\qquad
\langle\Psi_0| R_+^{q_+} R_-^{-q_-} |\Psi_0\rangle = \d_{q_+,0} \d_{q_-,0}
\quad
\forall q_+,q_- \in \mathbb{Z}
\label{Qs_and_Rs}
\ee
and commuting with all $\r_{r}(p)$ for $p \neq 0$; see, e.g., \cite{LaMo}.
For the latter, i.e., densities with non-zero momenta, we find it convenient to introduce boson creation and annihilation operators
\be
b^+_{p} = \left( b^-_{p} \right)^{\dagger},
\qquad
b^-_{p} = \begin{cases}
	-i\sqrt{\frac{2\pi}{L|p|}} \r_+(p) & \textnormal{if}\;\, p > 0, \\
	+i\sqrt{\frac{2\pi}{L|p|}} \r_-(p) & \textnormal{if}\;\, p < 0, \\
\end{cases}
\label{bpm}
\ee
respectively, satisfying
\be
\left[ b^-_{p\ppr},b^+_{p'} \right] = \d_{p,p'},
\qquad
\left[ b^\pm_{p\ppr}, b^\pm_{p'} \right] = 0,
\qquad
b^-_p|\Psi_0\rangle = 0
\quad
\forall p \neq 0.
\ee
It follows that the exact eigenstates of $H_{0}$ (the eigenvalues are given by \eqref{interacting_eigenvalue} for $\l = 0$) are
\be
|\Psi_{0,\mathbf{m}}\rangle
	= \left( \prod_{p \neq 0} \frac{\left(b^+_p\right)^{m(p)}}{\sqrt{m(p)!}} \right)
		R_+^{q_+} R_-^{-q_-} |\Psi_{0}\rangle
\label{non-interacting_eigenstate}
\ee
for $\mathbf{m} = \left\{ \left( m(p) \right)_{p \neq 0}, q_+, q_- \right\}$ with $m(p) \in \mathbb{N}$ and $q_+, q_- \in \mathbb{Z}$, where at most finitely many of the $m(p)$ are non-zero.
The ground state of $H_{0}$ is identified as $|\Psi_{0}\rangle = |\Psi_{0,\mathbf{0}}\rangle$.
These states form an orthonormal basis for the fermion Fock space $\mathcal{F}$, and we denote by $\mathcal{D}$ the set of all finite linear combinations of these eigenstates, from which $\mathcal{F}$ can be obtained by norm-completion; see, e.g., \cite{RS}.
The integer pair $(q_+,q_-)$ can be interpreted as the chiral charges of a given state, and $Q_\pm$ and $R_\pm$ as charge and charge-changing operators, respectively.
It follows that
\be
|\Psi_{\l,\mathbf{m}}\rangle
	= e^{-iS_{\l}} |\Psi_{0,\mathbf{m}}\rangle
\label{interacting_eigenstate}
\ee
are the exact eigenstates of $H_{\l}$ with eigenvalues
\be
\mathcal{E}_{\l,\mathbf{m}}
	= \frac{\pi}{L}\vF \left( q_+^2 + q_-^2
		+ \frac{\l\hat{\Vint}(0)}{\pi\vF} (q_+ + q_-)^2 \right)
	+ \sum_{p \neq 0} v_{\l}(p) |p| m(p) + W_{\l}.
\label{interacting_eigenvalue}
\ee
As for the case without interaction, the ground state of $H_{\l}$ is identified as $|\Psi_{\l}\rangle = |\Psi_{\l,\mathbf{0}}\rangle$.

We note that, even though many identities stated involve unbounded operators, they are well-defined on $\mathcal{D}$; see, e.g., \cite{LaMo, dWoLa} for mathematical details.


\section{Luttinger model with an external field: Proof of Theorem~\ref{Thm:H_lambda_mu}}
\label{Sec:Luttinger_model_with_an_external_field}
We extend the exact solution of the Luttinger model to include the case with an external field $\Wint(x)$ satisfying the condition $\int_{-L/2}^{L/2} dx\, \Wint(x) = 0$ as well as periodic boundary conditions (which we recall are needed if $L$ is finite).
In other words, we consider the Hamiltonian
\be
H_{\l,\m} = H_{0} + \l \Hint - \m\Pint
\label{hamiltonian_with_ext_field_in_Fourier_space}
\ee
with
\be
\Pint
= \sum_{r} \sum_{p} \frac{1}{L} \hat{\Wint}(-p) \r_{r}(p),
\ee
where $\hat{\Wint}(p)$ is the Fourier transform of $\Wint(x)$ satisfying $\hat{\Wint}(0) = 0$ due to the condition that the integral of $\Wint(x)$ over the entire space is zero.
To diagonalize $H_{\l,\m}$ we introduce an operator $A_{\l,\m}$ defined as
\be
A_{\l,\m}
= \sum_{r} \sum_{p \neq 0} \frac{1}{L} \hat{\h}(-p) r \r_{r}(p)
\ee
for some suitable function $\h(x)$ with Fourier transform $\hat{\h}(p)$, chosen so that it removes all terms linear in the densities; we note that, without loss of generality, we may set $\hat{\h}(0) = 0$.
It follows that $\left[ A_{\l,\m}, \r_{r}(p) \right] = - p \hat{\h}(p)/2\pi$, which implies
\be
e^{iA_{\l,\m}} \r_{r}(p) e^{-iA_{\l,\m}}	
= \r_{r}(p) - \frac{1}{2\pi} ip \hat{\h}(p)
\label{f1}
\ee
(to see this, define $g_{r}(u) = e^{iuA_{\l,\m}} \r_{r}(p) e^{-iuA_{\l,\m}}$, differentiate, and solve the resulting differential equation $g'_{r}(u) = - ip \hat{\h}(p)/2\pi$ with $g_{r}(0) = \r_{r}(p)$).
This operator applied to each term in the Hamiltonian yields:
\be
\begin{aligned}
e^{iA_{\l,\m}} H_{0} e^{-iA_{\l,\m}}
	& = H_{0} - \vF \sum_{r} \sum_{p} \frac{1}{L} (-ip) \hat{\h}(-p) \r_{r}(p)
	+ \frac{\vF}{2\pi} \sum_{p} \frac{1}{L} p^2 \left| \hat{\h}(p) \right|^2, \\
e^{iA_{\l,\m}} \Hint e^{-iA_{\l,\m}}
	& = \Hint - \vF \sum_{r} \sum_{p} \frac{1}{L}
		\frac{2\hat{\Vint}(p)}{\pi\vF} (-ip)\hat{\h}(-p)\r_{r}(p)
	+ \frac{\vF}{\pi} \sum_{p} \frac{1}{L}
		\frac{2\hat{\Vint}(p)}{\pi\vF} p^2 \left| \hat{\h}(p) \right|^2, \\
e^{iA_{\l,\m}} \Pint e^{-iA_{\l,\m}}
	& =\Pint - \frac{1}{\pi} \sum_{p} \frac{1}{L} \hat{\Wint}(-p) ip \hat{\h}(p),
\end{aligned}
\ee
where $\left| \hat{\h}(p) \right|^2 = \hat{\h}(-p)\hat{\h}(p)$.
Therefore, by choosing $\hat{\h}(p)$ such that
\be
\vF \left( 1 + {2\l\hat{\Vint}(p)}/{\pi\vF} \right) ip \hat{\h}(p) + \m \hat{\Wint}(p) = 0
\label{iph1}
\ee
is satisfied, we find that the Hamiltonian can be diagonalized using $e^{iS_{\l}}$ in \eqref{S_operator}:
\be
e^{iS_{\l}}e^{iA_{\l,\m}} H_{\l,\m} e^{-iA_{\l,\m}}e^{-iS_{\l}}
= H_{0} - T_{\l} + D_{\l} + W_{\l,\m}
\label{hamiltonian_with_ext_field_diagonalized}
\ee
with
\be
W_{\l,\m}
= W_\l - \frac{\m^2\vF}{2\pi} \sum_{p} \frac{1}{L}
	\frac{K_{\l}(p)^2}{v_{\l}(p)^2} \left| \hat{\Wint}(p) \right|^2
\ee
using $v_{\l}(p)$ in \eqref{vL}, $K_{\l}(p)$ in \eqref{Luttinger_parameter}, and $W_{\l}$ in \eqref{Diagonalized_Hamiltonain_ops}, where $|\hat{\Wint}(p)|^2 = \hat{\Wint}(-p)\hat{\Wint}(p)$.
We note that the Luttinger parameter $K_{\l}(p)$ can be expressed as
\be
K_{\l}(p) = e^{2\ph_\l(p)},
\label{Luttinger_parameter_ph}
\ee
which, using \eqref{S_operator}, gives precisely \eqref{Luttinger_parameter}.
(The relation in \eqref{Luttinger_parameter_ph} can be shown to hold for the general interactions considered in Appendix~\ref{App:General_interactions}.)
The ground state of $H_{\l,\m}$ is thus $|\Psi_{\l,\m} \rangle = e^{-iA_{\l,\m}} |\Psi_{\l}\rangle$ with ground state energy $W_{\l,\m}$.
Since the latter must be finite, this implies that the external field must satisfy
\be
\sum_{p} \frac{1}{L} \frac{K_{\l}(p)^2}{v_{\l}(p)^2} \left| \hat{\Wint}(p) \right|^2 < \io.
\label{mu_requirement}
\ee
This condition can be shown to be satisfied by $W(x)$ in \eqref{WxL}.
We also note that \eqref{iph1} can be written as
\be
ip \hat{\h}(p) = -\m \frac{K_{\l}(p)}{v_{\l}(p)} \hat{\Wint}(p)
\label{iph2}
\ee
using $v_{\l}(p)$ in \eqref{vL} and $K_{\l}(p)$ in \eqref{Luttinger_parameter}.
This together with \eqref{f1} yields the identity
\be
e^{iA_{\l,\m}} \r_{r}(p) e^{-iA_{\l,\m}}
= \r_{r}(p) + \frac{\m}{2\pi} \frac{K_{\l}(p)}{v_{\l}(p)} \hat{\Wint}(p)
\quad
\forall p,
\label{epiA_rsp_emiA}
\ee
which shows that the external field appears as a c-number added to $\r_{r}(p)$ after the transformation under $e^{iA_{\l,\m}}$.
Moreover, it follows from \eqref{interacting_eigenstate}, \eqref{interacting_eigenvalue}, and \eqref{hamiltonian_with_ext_field_diagonalized} that the exact eigenstates of $H_{\l,\m}$ are
\be
|\Psi_{\l,\m,\mathbf{m}}\rangle
= e^{-iA_{\l,\m}} e^{-iS_{\l}} |\Psi_{0,\mathbf{m}}\rangle
\label{H_lm_eigenstates}
\ee
for $\mathbf{m} = \left\{ \left( m(p) \right)_{p \neq 0}, q_+, q_- \right\}$ with $m(p) \in \mathbb{N}$ and $q_+, q_- \in \mathbb{Z}$, where at most finitely many of the $m(p)$ are non-zero.
The corresponding eigenvalues are
\be
\mathcal{E}_{\l,\m,\mathbf{m}}
= \frac{\pi}{L}\vF \left( q_+^2 + q_-^2
	+ \frac{\l\hat{\Vint}(0)}{\pi\vF} (q_+ + q_-)^2 \right)
	+ \sum_{p \neq 0} v_{\l}(p) |p| m(p) + W_{\l,\m}.
\label{H_lm_eigenvalues}
\ee
It is clear that the eigenstates of $H_{\lambda,\mu}$ in \eqref{H_lm_eigenstates} form a complete orthonormal basis for the fermion Fock space $\mathcal{F}$ defined in Sec.~\ref{Sec:Exact_solution_of_the_Luttinger_model}.
This implies that $H_{\l,\m}$ is self-adjoint on $\mathcal{F}$ and has a pure point spectrum. 
Moreover, it follows from \eqref{H_lm_eigenvalues} that all eigenvalues of $H_{\lambda,\mu}$ are $\geq W_{\lambda,\mu}$ and that its ground state can be identified as $|\Psi_{\l,\m}\rangle = |\Psi_{\l,\m,\mathbf{0}}\rangle$ which is clearly non-degenerate.
This also implies that $H_{\lambda,\mu}$ is bounded from below.
This proves Theorem~\ref{Thm:H_lambda_mu}.


\section{Evolution following a quench}
\label{Sec:Evolution_following_a_quench}
Assuming the system is in the ground state $|\Psi_{\l,\m} \rangle$ of $H_{\l,\m}$ in \eqref{hamiltonian_with_ext_field_in_Fourier_space}, we study its evolution under $H_{\l'}$ and compute the expectation values of certain local observables.
First, we consider the total density $\r(x) = \r_+(x) + \r_-(x)$ and the current $j(x) = \vF (\r_+(x) - \r_-(x))$.
Second, we compute the fermion two-point correlation function.
As we will see, the results in \eqref{epiS_rsp_emiS} and \eqref{epiA_rsp_emiA} are the workhorse in these computations.
We recall that the expectation values are first computed for finite $L$ and then we pass to the thermodynamic limit (cf.\ Remark~\ref{Remark:Thermodynamic_limit}).

\subsection{Total density and current: Proof of (\ref{Rxt}) and (\ref{Ixt})}
To compute the expectation values of $\r(x)$ and $j(x)$ with respect to $|\Psi_{\l,\m}^{\l'}(t)\rangle$ in \eqref{Psi_lp_l_m}, where we recall that $|\Psi_{\l,\m}\rangle = e^{-iA_{\l,\m}} e^{-iS_{\l}} |\Psi_{0}\rangle$, we first derive the time evolution of $\r_{r}(p)$ under $H_{\l'}$, do a similarity transformation under $e^{iS_{\l}} e^{iA_{\l,\m}}$, take the inverse Fourier transform, and then compute the expectation value with respect to $|\Psi_{0}\rangle$.
Through repeated use of \eqref{epiS_rsp_emiS} and \eqref{epiA_rsp_emiA} we find the following transformation rule for the densities:
\begin{multline}
e^{iS_{\l}} e^{iA_{\l,\m}}
	e^{iH_{\l'}t} \r_{r}(p) e^{-iH_{\l'}t}
e^{-iA_{\l,\m}} e^{-iS_{\l}} \\
= \frac{\m}{2\pi\vF} \frac{K_{\l}(p)}{v_{\l}(p)} \hat{\Wint}(p)
	\left( \vF \cos(pv_{\l'}(p)t) - ir v_{\l'}(p) \sin(pv_{\l'}(p)t) \right) \\
\begin{aligned}
& + \r_{r}(p) \left(
		  u^{++}_{\l,\l'}(p) e^{-ip r v_{\l'}(p)t}
		- u^{--}_{\l,\l'}(p) e^{ip r v_{\l'}(p)t}
	\right) \\
& - \r_{-r}(p) \left(
		  u^{+-}_{\l,\l'}(p) e^{-ip r v_{\l'}(p)t}
		- u^{-+}_{\l,\l'}(p) e^{ip r v_{\l'}(p)t}
	\right)
\end{aligned}
\label{transformed_densities}
\end{multline}
with the coefficients
\be
\begin{aligned}
u^{++}_{\l,\l'}(p) & = \cosh\ph_{\l'}(p) \cosh(\ph_{\l'}(p)-\ph_{\l}(p)),
\quad
	& u^{--}_{\l,\l'}(p) & = \sinh\ph_{\l'}(p) \sinh(\ph_{\l'}(p)-\ph_{\l}(p)), \\
u^{+-}_{\l,\l'}(p) & = \cosh\ph_{\l'}(p) \sinh(\ph_{\l'}(p)-\ph_{\l}(p)),
\quad
	& u^{-+}_{\l,\l'}(p) & = \sinh\ph_{\l'}(p) \cosh(\ph_{\l'}(p)-\ph_{\l}(p)).
\end{aligned}
\label{u_coefficients}
\ee
It thus follows, using the inverse Fourier transform $\r_{r}(x) = \sum_{p} (1/L) \r_{r}(p) e^{ipx}$ and computing the expectation value of \eqref{transformed_densities} with respect to $|\Psi_{0}\rangle$, that the time evolution of the total density is
\be
\langle\Psi_{\l,\m}^{\l'}(t)| \r(x) |\Psi_{\l,\m}^{\l'}(t)\rangle_{L}
= \frac{\m}{2\pi} \sum_{p} \frac{1}{L}
	\frac{K_{\l}(p)}{v_{\l}(p)} \hat{\Wint}(p)
	2 \cos(pv_{\l'}(p)t) e^{ipx},
\label{Rxt_L}
\ee
and, similarly, that the time evolution of the current is
\be
\langle\Psi_{\l,\m}^{\l'}(t)| j(x) |\Psi_{\l,\m}^{\l'}(t)\rangle_{L}
= \frac{\m}{2\pi} \sum_{p} \frac{1}{L}
	\frac{K_{\l}(p)}{v_{\l}(p)} \hat{\Wint}(p) v_{\l'}(p)
	(-2i \sin(pv_{\l'}(p)t)) e^{ipx}.
\label{Ixt_L}
\ee
The results in \eqref{Rxt} and \eqref{Ixt} are obtained from \eqref{Rxt_L} and \eqref{Ixt_L} in the thermodynamic limit.

\subsection{Two-point correlation function: Proof of (\ref{psi_psi_factoring_infinite_L})--(\ref{Srxyt})}
To compute the two-point correlation function we use the following lemma which allows us to obtain the fermion fields $\psi^\pm_{r}(x)$ as limits of regularized fields (for proof see, e.g., \cite{CaHu} or Proposition~3.7 in \cite{LaMo}):

\begin{lemma}
\label{Lemma:Psi_reg}
Let $\e > 0$ and consider
\be
\psi^\pm_{r}(x;\e)
= L^{-1/2}
	e^{\mp i\pi rx Q_{r}/L} R^{\pm r}_{r} e^{\mp i\pi rx Q_{r}/L}
	\exp\left( \sum_{p > 0}\frac{\pi}{Lp} e^{- \e p} \right)
	\exp\left( \mp r\sum_{p \neq 0}\frac{2\pi}{Lp} \r_{r}(p) e^{ipx - \e |p|/2} \right)
\label{vertex_rep}
\ee
with $Q_{r}$ and $R_{r}$ defined in Sec.~\ref{Sec:Exact_solution_of_the_Luttinger_model}.
Then $\psi^\pm_{r}(x;\e)$ converge to $\psi^\pm_{r}(x)$ as $\e \to 0^+$ in the following distributional sense:
\be
\int_{-L/2}^{L/2} dx\, \psi^\pm_{r}(x) e^{\pm ikx}
= \lim_{\e \to 0^+} \int_{-L/2}^{L/2} dx\, \psi^\pm_{r}(x;\e) e^{\pm ikx}
\label{reg_psi_to_psi}
\ee
for all $k = \pi(2n+1)/L$ with $n \in \mathbb{N}$, where the limit on the right-hand side is in the strong sense \textnormal{\cite{RS}} on the domain $\mathcal{D}$ defined in Sec.~\ref{Sec:Exact_solution_of_the_Luttinger_model}.
\end{lemma}

With these fields, and by again repeatedly using \eqref{epiS_rsp_emiS} and \eqref{epiA_rsp_emiA}, we find, similar to \eqref{transformed_densities},
\begin{multline}
e^{iS_{\l}} e^{iA_{\l,\m}}
	e^{iH_{\l'}t} \psi^\pm_{r}(x;\e) e^{-iH_{\l'}t}
e^{-iA_{\l,\m}} e^{-iS_{\l}} \\
= \exp\left( \mp ir \frac{\m}{\vF} \sum_{p \neq 0} \frac{1}{L}
		\frac{K_{\l}(p)}{v_{\l}(p)} \hat{\Wint}(p)
		\left(
			\vF \cos(pv_{\l'}(p)t) - ir v_{\l'}(p) \sin(pv_{\l'}(p)t)
		\right) \frac{e^{ipx - \e |p|/2}}{ip}
	\right) \\
\begin{aligned}
& \times L^{-1/2}
		e^{\mp i\pi rx Q_{r}/L}
		e^{iH_{\l'}t} R_{r}^{\pm r} e^{-iH_{\l'}t}
		e^{\mp i\pi rx Q_{r}/L}
	\exp\left( \sum_{p > 0}\frac{\pi}{Lp} e^{- \e p} \right) \\
& \times \exp\left( \mp r\sum_{p \neq 0}\frac{2\pi}{Lp}
	\r_{r}(p) \left(
		  u^{++}_{\l,\l'}(p) e^{-ip r v_{\l'}(p)t}
		- u^{--}_{\l,\l'}(p) e^{ip r v_{\l'}(p)t}
	\right) e^{ipx - \e |p|/2}
	\right) \\
& \times \exp\left( \pm r\sum_{p \neq 0}\frac{2\pi}{Lp}
	\r_{-r}(p) \left(
		  u^{+-}_{\l,\l'}(p) e^{-ip r v_{\l'}(p)t}
		- u^{-+}_{\l,\l'}(p) e^{ip r v_{\l'}(p)t}
	\right) e^{ipx - \e |p|/2}
	\right)
\end{aligned}
\label{psi_h0}
\end{multline}
with
\be
e^{iH_{\l'}t} R_{r}^{\pm r} e^{-iH_{\l'}t}
= e^{\pm i\pi v_{\l'}t \sum_{s=\pm} (K_{\l'}^{-1}+sK_{\l'}^{\vphantom{-1}})Q_{sr}/2L}
R_{r}^{\pm r}
e^{\pm i\pi v_{\l'}t \sum_{s=\pm} (K_{\l'}^{-1}+sK_{\l'}^{\vphantom{-1}})Q_{sr}/2L}
\ee
and the coefficients in \eqref{u_coefficients}.
It follows that the two-point correlation function is
\be
\langle\Psi_{\l,\m}^{\l'}(t)| \psi^+_{r}(x) \psi^-_{r}(y) |\Psi_{\l,\m}^{\l'}(t)\rangle_{L}
= e^{-ir\vF^{-1} A_{r,L}(x,y,t) (x-y)} S_{r,L}(x,y,t)
\label{psi_psi_factoring_finite_L}
\ee
with
\be
A_{r,L}(x,y,t)
= \m \sum_{p \neq 0}\frac{1}{L}
		\frac{K_{\l}(p)}{v_{\l}(p)} \hat{\Wint}(p)
		\left(
			\vF \cos(pv_{\l'}(p)t) - ir v_{\l'}(p) \sin(pv_{\l'}(p)t)
		\right) \frac{e^{ipx} - e^{ipy}}{ip(x-y)}
\label{Arxyt_L}
\ee
and
\be
S_{r,L}(x,y,t)
= \langle\Psi_{\l,0}^{\l'}(t)| \psi^+_{r}(x) \psi^-_{r}(y) |\Psi_{\l,0}^{\l'}(t)\rangle_{L}.
\label{Srxyt_L}
\ee
The results in \eqref{psi_psi_factoring_infinite_L}--\eqref{Srxyt} are obtained from \eqref{psi_psi_factoring_finite_L}--\eqref{Srxyt_L} in the thermodynamic limit.

\subsection{Two-point correlation function: Proof of (\ref{l_lp_two_point_func})--(\ref{eta_l_lp_equal_eta_l_gamma_l_lp})}
For the proof we use the regularized operators $\psi^q_{r}(x;\e)$ in Lemma~\ref{Lemma:Psi_reg} for $q = \pm$ and $\e > 0$.
It follows by setting $\m = 0$ in \eqref{psi_h0} that
\begin{multline}
e^{iS_{\l}}e^{iH_{\l'} t} \psi^q_{r}(x;\e) e^{-iH_{\l'} t}e^{-iS_{\l}} \\
	= L^{-1/2} e^{-i\pi qrx Q_{r}/L} e^{iH_{\l'}t} R_{r}^{qr} e^{-iH_{\l'}t} e^{-i\pi qrx Q_{r}/L}
		Z_{a,\e,L}(t)
		W_{1,r}^{q} (x,t;\e)
		W_{2,-r}^{q} (x,t;\e)
\label{vertex_op_rep}
\end{multline}
with the time-dependent factor
\be
Z_{a,\e,L}(t)
	= \exp \left( 
			-\sum_{p>0} \frac{\pi}{L}
			\frac{\h_{\l,\l'}(p) - \g_{\l,\l'}(p) \cos (2p v_{\l'}(p)t)}{p}
			e^{-\e p}
		\right)
\label{Z_factor}
\ee
(the dependence on the interaction range $a > 0$ is commented on further below) and so-called \emph{vertex operators},
\be
\begin{aligned}
W_{1,r}^{q} (x,t;\e)
	& = \bosonnoord{ \exp \biggl(
			-qr \sum_{p \neq 0} \frac{2\pi}{Lp} \r_{r}(p) \Bigl(
				  u^{++}_{\l,\l'}(p) e^{ip(x-r v_{\l'}(p)t)}
				- u^{--}_{\l,\l'}(p) e^{ip(x+r v_{\l'}(p)t)}
			\Bigr) e^{-\e|p|/2}
		\biggr) }, \\
W_{2,-r}^{q} (x,t;\e)
	& = \bosonnoord{ \exp \biggl(
			qr \sum_{p \neq 0} \frac{2\pi}{Lp} \r_{-r}(p) \Bigl(
				  u^{+-}_{\l,\l'}(p) e^{ip(x-r v_{\l'}(p)t)}
				- u^{-+}_{\l,\l'}(p) e^{ip(x+r v_{\l'}(p)t)}
			\Bigr) e^{-\e|p|/2}
		\biggr) },
\end{aligned}
\ee
where $\bosonnoord{\cdots}$ denotes \emph{boson normal ordering}, i.e., all boson creation operators are placed to the left of all boson annihilation operators; see, e.g., \cite{LaMo} for precise definitions.
Using the coefficients in \eqref{u_coefficients} one finds that $\h_{\l,\l'}(p)$ and $\g_{\l,\l'}(p)$ are given by
\be
\begin{aligned}
\h_{\l,\l'}(p)
& = 2 \left( u^{++}_{\l,\l'}(p)^2 + u^{--}_{\l,\l'}(p)^2 - 1 \right) \\
& = 2 \left( u^{+-}_{\l,\l'}(p)^2 + u^{-+}_{\l,\l'}(p)^2 \right) \\
& = \cosh 2\ph_{\l'}(p) \cosh 2(\ph_{\l'}(p)-\ph_{\l}(p)) - 1
\end{aligned}
\label{eta_ph_php}
\ee
and
\be
\begin{aligned}
\g_{\l,\l'}(p)
& = 4 u^{++}_{\l,\l'}(p) u^{--}_{\l,\l'}(p) \\
& = 4 u^{+-}_{\l,\l'}(p) u^{-+}_{\l,\l'}(p) \\
& = \sinh 2\ph_{\l'}(p) \sinh 2(\ph_{\l'}(p)-\ph_{\l}(p)),
\end{aligned}
\label{gamma_ph_php}
\ee
respectively.
Note that the dependence in \eqref{Z_factor} on the interaction range $a > 0$ is indicated to emphasize that $Z_{a,\e,L}(t)$ vanishes in the local limit $a \to 0^+$; see, e.g., \cite{LaMo} for details.

The technical parts of the computations are omitted here since they are identical to those in \cite{LaMo}, where they are explained in detail.
We mention only that it follows from \eqref{vertex_op_rep}--\eqref{gamma_ph_php} and Proposition~3.4 in \cite{LaMo} that
\begin{multline}
	Z_{a,\e,L}(t)^{2} \langle\Psi_0| W_{1,r}^{+} (x,t;\e) W_{2,-r}^{+} (x,t;\e)
		W_{1,r}^{-} (y,t;\e) W_{2,-r}^{-} (y,t;\e) |\Psi_0\rangle_L \\
	= \exp \left(
			\sum_{p>0} \frac{2\pi}{L} \left(
			\frac{e^{ipr(x-y)}}{p}
			+ \frac{\h_{\l,\l'}(p) - \g_{\l,\l'}(p) \cos (2p v_{\l'}(p)t)}{p}
			\left( \cos p(x-y) - 1 \right)
			\right)
			e^{-\e p}
		\right),
\label{result_from_Prop_3_4_in_LaMo}
\end{multline}
where the only factor which does not depend on $\l$ or $\l'$ corresponds to
\be
\langle\Psi_{0,0}^{0}(t)| \psi^+_{r}(x) \psi^-_{r}(y) |\Psi_{0,0}^{0}(t)\rangle
= \lim_{\e \to 0^+} \lim_{L \to \io}
	\frac{1}{L} \exp \left( \sum_{p>0} \frac{2\pi}{L}
		\frac{e^{ipr(x-y)}}{p} e^{-\e p} \right)
= \frac{i}{2\pi r(x-y) + i0^+}.
\label{psi_p_psi_m_non-interacting}
\ee
The latter is precisely the two-point correlation function in the thermodynamic limit when there is no interaction (and no domain wall).
Moreover, the operators $Q_{r}$ and $R_{r}$ in \eqref{vertex_op_rep} do not contribute if we take $L \to \io$ other than that the last identity in \eqref{Qs_and_Rs} puts constraints on which two-point correlation functions are non-zero: the ground state expectation value of $\psi^{q\ppr}_{r\ppr}(x)\psi^{q'}_{r'}(y)$ for $r, r' = \pm$ and $q, q' = \pm$ can be non-zero only if $r = r'$ and $q = -q'$.
Therefore, after taking the limit $\e \to 0^+$, the only non-zero two-point correlation function in the thermodynamic limit is the one given in \eqref{l_lp_two_point_func} with $\h_{\l,\l'}(p)$ and $\g_{\l,\l'}(p)$ in \eqref{eta_gamma_K_Kp}; this follows from \eqref{Luttinger_parameter_ph} and \eqref{vertex_op_rep}--\eqref{psi_p_psi_m_non-interacting}.
Finally, we note that the identity in \eqref{eta_l_lp_equal_eta_l_gamma_l_lp} follows from \eqref{eta_ph_php} and \eqref{gamma_ph_php} using $\h_{\l}(p) = \h_{\l,\l}(p) = \cosh 2\ph_{\l}(p) - 1$.


\section{Approach to steady state: Proof of Theorem~\ref{Thm:Asymptotics}}
\label{Sec:Approach_to_steady_state}
The asymptotic behaviors of $I(x,t)$, $R(x,t)$, and $A_{r}(x,y,t)$ can be studied by considering the integrals in \eqref{Rxt}, \eqref{Ixt}, and \eqref{Arxyt}, respectively.
In Appendix~\ref{App:Asymptotics} we show that
\be
|R(x,t) - R| \le \frac{C_R}{t},
\qquad
|I(x,t) - I| \le \frac{C_I}{t},
\qquad
|A_{\pm}(x,y,t) - A_{\pm}| \le \frac{C_A}{t}
\quad
\forall x,y
\label{R_I_A_asymptotics}
\ee
for certain finite constants $C_R, C_I, C_A > 0$ with
\be
R = 0,
\qquad
I = \frac{\m}{2\pi} \frac{K_{\l}v_{\l'}}{v_{\l}},
\qquad
A_{\pm}
= \pm \frac{\m}{2} \frac{K_{\l}v_{\l'}}{v_{\l}},
\label{R0_I0_Ar0}
\ee
which are manifestly translation invariant.
From this the result in \eqref{Rxt_Ixt_Arxyt_Asymptotics} follows.
Similarly, the asymptotic behavior of $S_{r}(x,y,t)$ can be studied by considering the integral in \eqref{l_lp_two_point_func}.
In Appendix~\ref{App:Asymptotics} we prove that 
\be
\lim_{t \to \io} S_{r}(x,y,t)
	= \frac{i}{2\pi r (x-y) + i0^+}
	\exp \left( \int_0^\io dp \frac{\h_{\l,\l'}(p)}{p}
		\left( \cos p(x-y) - 1 \right) \right),
\label{l_lp_two_point_func_t_to_io}
\ee
where the contribution containing $\g_{\l,\l'}(p)$, which is non-zero if $\l \neq \l' \neq 0$, has disappeared in the limit $t \to \io$.
The result in \eqref{acc} follows from \eqref{l_lp_two_point_func_t_to_io} using \eqref{psi_psi_factoring_infinite_L} and \eqref{Rxt_Ixt_Arxyt_Asymptotics}.
This completes the proof of Theorem~\ref{Thm:Asymptotics}.

A result similar to \eqref{l_lp_two_point_func_t_to_io} can be derived for the non-equal-time two-point correlation function $\langle\Psi_{\l,0}^{\l'}(t)| \psi^+_{r}(x_+,t_+) \psi^-_{r}(x_-,t_-) |\Psi_{\l,0}^{\l'}(t)\rangle$ for fields $\psi^\pm_{r}(x_\pm,t_\pm) = e^{iH_{\l'}t_\pm} \psi^\pm_{r}(x_\pm) e^{-iH_{\l'}t_\pm}$ with $t_\pm \ll t$ using the methods in Sec.~\ref{Sec:Evolution_following_a_quench}.
Although we omit the explicit expression we note that it shares the property with \eqref{l_lp_two_point_func_t_to_io} that contributions containing $\g_{\l,\l'}(p)$ disappear as $t \to \io$.
Moreover, these particular contributions are the only ones which are not time-translation invariant, one example is a dependence on $t_+ + t_-$, and, since they disappear, time-translation invariance is recovered asymptotically in time.
One could naively try to reproduce this by treating the final state as a ground state of some Luttinger Hamiltonian (cf.\ Sec.~\ref{Sec:Luttinger_model_with_constant_chemical_potentials}).
However, for $\l \neq \l' \neq 0$, the same non-equal-time two-point correlation function for such a ground state of a Luttinger Hamiltonian, which cannot be $H_{\l'}$ since $\h_{\l'}(p) \neq \h_{\l,\l'}(p)$, depends, in general, on $t_+ + t_-$.
Therefore, since such a state is not time-translation invariant, it cannot correspond to the final steady state when $\l \neq \l' \neq 0$; cf.\ also \cite{IC}.
On the other hand, we show in Sec.~\ref{Sec:Luttinger_model_with_constant_chemical_potentials} that for $\l = \l'$ the two-point correlation function supports such a description, but with different chemical potentials for right- and left-moving fermions if $\mu \neq 0$.
Moreover, since $\g_{\l,\l'}(p) = 0$ if $\l = \l'$, this is also true for the non-equal-time correlation function.


\section{Luttinger model with constant chemical potentials: Proof of Theorem~\ref{Thm:Hamiltonian_2}}
\label{Sec:Luttinger_model_with_constant_chemical_potentials}
We now show that, for $\l = \l'$, the steady state, obtained asymptotically in time, has the same equal-time two-point correlation function in the thermodynamic limit as the ground state of the Hamiltonian in \eqref{hamiltonian2} with constant chemical potentials $\m_\pm$ satisfying $\m_+ + \m_- = 2\m_0$.
We again recall that the model is defined for finite $L$ (cf.\ Remark~\ref{Remark:Thermodynamic_limit}).

To diagonalize the Hamiltonian $H$ in \eqref{hamiltonian2} we use so-called \emph{large gauge transformations}, which are implemented by the unitary operators $R_{r}$ defined in Sec.~\ref{Sec:Exact_solution_of_the_Luttinger_model}:
\be
\begin{aligned}
R_{r}^{-r w_{r}} \psi^{\pm}_{r'}(x) R_{r}^{r w_{r}}
& = e^{\mp ir \d_{r,r'} w_r 2\pi x/L } \psi^{\pm}_{r'}(x), \\
R_{r}^{-r w_{r}} \r_{r'}(x) R_{r}^{r w_{r}}
& = \r_{r'}(x) + \d_{r,r'} w_{r}/L
\end{aligned}
\label{large_gauge_transf}
\ee
for $w_{r} \in \mathbb{Z}$; see, e.g., \cite{LaMo}.
We note that the requirement that the $w_{r}$ are integers constrain the possible values of $\m_{r} - \m_0$ for which the following arguments apply, but this will be of no consequence in the thermodynamic limit and can thus be ignored.
For the densities, the transformation in \eqref{large_gauge_transf} implies that there is a shift in the zero-mode contribution, which we will see corresponds to a shift in the ground state.
It follows from Sec.~\ref{Sec:Exact_solution_of_the_Luttinger_model} that
\be
e^{iS_{\l}} H e^{-iS_{\l}}
= H_{0} - T_{\l} + D_{\l} + W_{\l} - \sum_{r} (\m_{r} - \m_0) Q_{r}
\ee
for $H$ in \eqref{hamiltonian2}, and therefore
\begin{multline}
R_+^{-w_+} R_-^{w_-} e^{iS_{\l}} H e^{-iS_{\l}} R_+^{w_+} R_-^{-w_-} \\
= \frac{\pi}{L} \vF \left( \sum_{r} (Q_{r} + w_{r})^2
+ \frac{\l \hat{\Vint}(0)}{\pi\vF} \left( \sum_{r} Q_{r} + w_{r} \right)^2 \right)
+ D_{\l} + W_{\l} - \sum_{r} (\m_{r} - \m_0) (Q_{r} + w_{r}).
\end{multline}
To make this diagonal we must choose $w_\pm = L(\m_\pm - \m_0)/2\pi\vF$, which yields
\be
\begin{aligned}
R_+^{-w_+} R_-^{w_-} \psi^{\pm}_{r}(x) R_+^{w_+} R_-^{-w_-}
	& = e^{\mp ir\vF^{-1} (\m_{r} - \m_0) x} \psi^{\pm}_{r}(x), \\
R_+^{-w_+} R_-^{w_-} \r_{r}(x) R_+^{w_+} R_-^{-w_-}
	& = \r_{r}(x) + \frac{\m_{r} - \m_0}{2\pi\vF},
\end{aligned}
\label{large_gauge_transf_w_fixed}
\ee
and
\be
R_+^{-w_+} R_-^{w_-} e^{iS_{\l}} H e^{-iS_{\l}} R_+^{w_+} R_-^{-w_-}
= H_{0} - T_{\l} + D_{\l} + W_{\l}
	- \frac{L}{4\pi\vF} \sum_{r} (\m_{r} - \m_0)^2
\label{large_gauge_transf_w_fixed_Hamiltonian}
\ee
(where we have used the condition $\m_+ + \m_- = 2\m_0$).
This means that the ground state of $H$ is $|\Psi\rangle = R_+^{w_+} R_-^{-w_-} |\Psi_\l\rangle$ with ground state energy $W_{\l} - {L} (\m_+ - \m_-)^2 / {8\pi\vF}$.
We note that $|\Psi\rangle$ can also be interpreted as an excited eigenstate of $H_{\l}$ since it has chiral charges $w_\pm \neq 0$ if $\m_\pm \neq \m_0$ (cf.\ \eqref{non-interacting_eigenstate}).
Using the results in Sec.~\ref{Sec:Exact_solution_of_the_Luttinger_model} we can construct all eigenstates and eigenvalues of $H$ for finite $L$, but we omit the details since they are similar to \eqref{non-interacting_eigenstate}--\eqref{interacting_eigenvalue}.
This proves the statements about $H$ in Theorem~\ref{Thm:Hamiltonian_2}; cf.\ the end of Sec.~\ref{Sec:Luttinger_model_with_an_external_field}.
Moreover, as explained in Sec.~\ref{Sec:Evolution_following_a_quench}, we can compute the two-point correlation function $\langle\Psi| \psi^+_{r}(x) \psi^-_{r}(y) |\Psi\rangle_L$ using the first identity in \eqref{large_gauge_transf_w_fixed}, which in the thermodynamic limit gives the result in \eqref{psi_psi_correl_2}, but again we omit the details since they are similar to \eqref{vertex_op_rep}--\eqref{psi_p_psi_m_non-interacting}.
This completes the proof of Theorem~\ref{Thm:Hamiltonian_2}.

Lastly, we note that the identities in \eqref{psi_tilde_to_psi} relating the fields $\psi^\pm_{r}(x)$ and $\tilde\psi^\pm_{r}(x)$ can be proven in the same way as \eqref{large_gauge_transf_w_fixed}.
The only difference is in the transformation of the corresponding Hamiltonians: it follows from \eqref{hamiltonian1} and \eqref{psi_tilde_to_psi} that there must be counterterms in \eqref{hamiltonian}, not present in \eqref{hamiltonian1}, for the Hamiltonians expressed using $\tilde\psi^\pm_{r}(x)$ and $\psi^\pm_{r}(x)$ to be the same; cf.\ also \cite{BGM}.


\section{Concluding remarks}
\label{Sec:Concluding_remarks}
We studied properties of the Luttinger model with a non-local interaction following a quench from a domain wall initial state.
The evolution of the local observables we consider is ballistic, but, at the same time, dispersive if $\l' \neq 0$.
The dispersion effects appear as oscillations traveling ahead of the wave fronts as in Figs.~\ref{Fig:Total_density}~and~\ref{Fig:Current}.
This is in agreement with analytical and numerical results in \cite{SM, LaMi} for quantum XXZ spin chains.
However, it should be noted that, since the model in \cite{LaMi} is mapped to the Luttinger model with a delta-function interaction, the evolution in that case is non-dispersive, similar to our non-interacting case, and there are divergences which are absent in our results.
Asymptotically in time, the system reaches a steady state that retains memory of the initial state, meaning there is equilibration (stabilization).
This was previously found in \cite{LaMi} for quantum XXZ spin chains and in \cite{RDYO} for non-interacting bosons.
We also showed that, if $\l \neq \l' \neq 0$, the final steady state has exponents which are different from those in equilibrium.
This generalizes previous results in \cite{C, IC, MW} to more general interaction quenches.

The final state is clearly different from the ground state of $H_{\l'}$ in general (i.e., if $\l \neq \l' \neq 0$ or $\m \neq 0$) and thus cannot be described as a thermal state obtained from the usual canonical ensemble.
Since the system is integrable this is not surprising: thermalization in the usual sense is not expected but rather in the sense of a generalized canonical ensemble; see, e.g., \cite{RDYO, IC}.
For $\l = \l'$ we showed that the (equal-time) two-point correlation function for the final state is equal to that for the ground state of the Hamiltonian $H$ in \eqref{hamiltonian2} with different chemical potentials for right- and left-moving fermions (this can also be verified for the non-equal-time correlation function).
If true for all $N$-point correlation functions, this would suggest that the Gibbs measure of the final state is $e^{-\b H}$ with $H$ in \eqref{hamiltonian2} in the zero-temperature limit $\b \to \io$.
Therefore, when $\l = \l'$, the suggested generalized canonical ensemble needed to describe the final state consists of the conserved quantities $H_{\l'}$ and $Q_\pm$.
However, for $\l \neq \l' \neq 0$ we showed that this particular generalized canonical ensemble is too simple.
To test whether a generalized canonical ensemble can describe the final steady state, when $\l \neq \l' \neq 0$, additional conserved quantities (which are higher-order polynomials in the commuting operators $n^{\vphantom{+}}_p = b^+_p b^-_p$ with $b^\pm_p$ in \eqref{bpm}) would have to be included.
This problem is left open.

We stress that universality of conductance is found via a fully non-equilibrium approach.
Previous explanations of this universality were in a near-to-equilibrium setting; see, e.g., \cite{MS, K, ACF1, ACF2}.
Whether the conductance in the Luttinger model is renormalized or not has been debated. 
Earlier works found that it should be renormalized by the interaction \cite{KF}, but later experimental and theoretical results found it to be universal; see, e.g., \cite{ACF1}.
We found that these different results can be reconciled in a dynamical approach, which we believe sheds new light on this issue.
The key observation is that the current and the chemical potential difference are renormalized by the interaction, but that the renormalizations cancel when computing the Landauer conductance \cite{La2, ACF2}.
On the other hand, if one uses a non-renormalized chemical potential difference when defining the conductance, the renormalizations do not cancel.
This is summarized by the relation in \eqref{I_Gll_mLMR_G_mpmm}, which, in particular, shows that $G = I/(\m_+ - \m_+) = 1/2\pi$ (in units where $e = \hbar = 1$) is always universal.
This relation also makes clear that our results for the Luttinger model are not in conflict with the non-universal results in \cite{SM} for quantum XXZ spin chains.
Indeed, setting $\l = \l'$ in \eqref{Rxt_Ixt_Arxyt_Asymptotics} yields the non-universal value $G_{\l,\l} = {I}/{\m} = {K_{\l}}/{2\pi}$ in \cite{KF}, which is in agreement with the result in \cite{SM} if one takes spin degrees of freedom into account.

It is natural to conjecture that universality of conductance persists even when no exact solution is available.
In equilibrium, renormalization group techniques can explain universality also in cases without exact solutions.
It would be interesting to extend such an approach to a non-equilibrium setting.
We also note a connection between the Landauer conductance and the Thouless conductance for reflectionless transport of free fermions studied in \cite{BJLP}.
It would be interesting to investigate if the mechanism explored in \cite{BJLP} is related to our universality result.
Finally, we mention that we expect our method to be applicable to more general states, such as, e.g., thermal states.

\paragraph*{Acknowledgements:}
We would like to thank Natan Andrei, J\"{u}rg Fr\"{o}hlich, Giovanni Gallavotti, Krzysztof Gaw\k{e}dski, and P\"{a}ivi T\"{o}rm\"{a} for helpful and inspiring discussions.
E.~L.\ acknowledges financial support by the G\"{o}ran Gustafsson Foundation (GGS 1221).
The research by J.~L.~L.\ was supported by NSF Grant No.\ DMR 1104500 and by AFOSR Grant No.\ FA9550-16.
V.~M.\ acknowledges support from the  PRIN-MIUR National Grant ``Geometric and analytic theory of Hamiltonian systems in finite and infinite dimensions.''
P.~M.\ is thankful for financial support from the Olle Eriksson Foundation for Materials Engineering (No.\ VT-2015-0001).


\begin{appendix}


\section{Asymptotics}
\label{App:Asymptotics}
To study the asymptotic behavior of $R(x,t)$, $I(x,t)$, $A_{r}(x,y,t)$, and $S_{r}(x,y,t)$ in Sec.~\ref{Sec:Approach_to_steady_state} we use the following lemma (the proof is given at the end of this appendix):

\begin{lemma}
\label{Lemma:Asymptotics}
Let
\be
F(t) = \strokedint_{-\io}^{\io} \frac{dp}{2\pi} \frac{f(p)}{p} e^{ip u(p)t},
\qquad
F_0(t) = \strokedint_{-\io}^{\io} \frac{dp}{2\pi} \frac{f(0)}{p} e^{ip u(0)t}
\label{FF0}
\ee
(interpreted as Cauchy principal values) with real-valued functions $f(p)$ and $u(p)$ (for $p \in \mathbb{R}$) satisfying the following conditions:
\be
\begin{aligned}
& 1. \quad (f(p) - f(0)) p^{-1} \in C^{1}(\mathbb{R})\;\, \textnormal{(a.e.)}, \\
& 2. \quad (f(p) - f(0)) p^{-1},\; d((f(p) - f(0)) p^{-1})/dp \in L^{1}(\mathbb{R}), \\
& 3. \quad (f(p) - f(0))p^{-1} \to 0 \;\, \textnormal{as} \;\, p \to \pm \io,
\end{aligned}
\label{f_conditions}
\ee
and
\be
\begin{aligned}
& 1. \quad u(p) \in C^{2}(\mathbb{R}) \;\, \textnormal{(a.e.)},\\
& 2. \quad u(p),\;
		du(p)/dp,\;
		d^2u(p)/dp^2 \in L^{1}(\mathbb{R}),\\
& 3. \quad u^{\textnormal{g}}(p) = {d(u(p)p)}/{dp} > 0 \quad \forall p.
\end{aligned}
\label{u_conditions}
\ee
Then
\be
|F(t) - F_0(t)| \le \frac{C}{t}
\label{FF0Ct}
\ee
for some finite constant $C > 0$.
\end{lemma}

\begin{proof}[Proof of \eqref{R_I_A_asymptotics} and \eqref{R0_I0_Ar0}]
If $\Wint(x) = 1/2 - \th(x)$, or our regularized version thereof, the Fourier transform of the external field is $\hat{\Wint}(p) = i/p$ or tends to this as $p \to 0$.
Letting $u(p) = v_{\l'}(p)$ and choosing $f(p)$ as the function corresponding to $R(x,t)$ in \eqref{Rxt}, $I(x,t)$ in \eqref{Ixt}, and $A_{r}(x,y,t)$ in \eqref{Arxyt}, respectively, the conditions in \eqref{V_conditions_2} imply that the conditions in \eqref{f_conditions} and \eqref{u_conditions} are satisfied; this follows using \eqref{vL} and \eqref{Luttinger_parameter}.
We now note that
\be
\strokedint_{-\io}^{\io} \frac{dp}{2\pi} \frac{1}{p} \cos(p u(0)t) = 0,
\qquad
\strokedint_{-\io}^{\io} \frac{dp}{2\pi} \frac{1}{p} \sin(p u(0)t) = \frac{1}{2}
\quad
\forall t > 0.
\label{Asymptotics_lemma_helper}
\ee
The identities in \eqref{R_I_A_asymptotics} and \eqref{R0_I0_Ar0} then follow from Lemma~\ref{Lemma:Asymptotics} and \eqref{Asymptotics_lemma_helper} by splitting $e^{ip v_{\l'}(p)t}$ into real and imaginary parts.
\end{proof}

\begin{proof}[Proof of \eqref{l_lp_two_point_func_t_to_io}]
We set $u(p) = 2v_{\l'}(p)$ and $f(p) = 2\pi \g_{\l,\l'}(p) \left( \cos p(x-y) - 1 \right)$ in Lemma~\ref{Lemma:Asymptotics} and note that the conditions in \eqref{V_conditions_2} imply that the conditions in \eqref{f_conditions} and \eqref{u_conditions} are satisfied; this follows using \eqref{vL}, \eqref{Luttinger_parameter}, and $\g_{\l,\l'}(p)$ in \eqref{eta_gamma_K_Kp}.
It then follows from Lemma~\ref{Lemma:Asymptotics} and \eqref{Asymptotics_lemma_helper}, by splitting $e^{2ip v_{\l'}(p)t}$ into real and imaginary parts, that the time-dependence in the integral in \eqref{l_lp_two_point_func} disappears as $t \to \io$.
\end{proof}

\begin{proof}[Proof of Lemma~\ref{Lemma:Asymptotics}]
Define $\tp$ by the change of variables
\be
u(0)\tp = u(p)p,
\qquad
u(0) d\tp = u^{\textnormal{g}}(p) dp.
\ee
Also, define $\tf(\tp)$ by
\be
\frac{\tf(\tp) d\tp}{\tp} = \frac{f(p) dp}{p},
\ee
which implies
\be
\tf(\tp) = f(p) \frac{u(p)}{u^{\textnormal{g}}(p)},
\qquad
\tf(0) = f(0).
\ee
By relabeling $p$ by $\tp$ in $F_0(t)$, we find
\be
F(t) - F_0(t)
= \strokedint_{-\io}^{\io} \frac{d\tp}{2\pi} \frac{\tf(\tp) - \tf(0)}{\tp} e^{i\tp u(0)t},
\ee
and using integration by parts,
\be
F(t) - F_0(t)
= \left[ \frac{1}{2\pi} \frac{\tf(\tp) - \tf(0)}{\tp} \frac{-ie^{i\tp u(0)t}}{u(0)t} \right]_{-\io}^{\io}
+ \strokedint_{-\io}^{\io} \frac{d\tp}{2\pi}
	\frac{d}{d\tp} \left( \frac{\tf(\tp) - \tf(0)}{\tp} \right)
	\frac{ie^{i\tp u(0)t}}{u(0)t}.
\ee
This implies that \eqref{FF0Ct} holds with
\be
C = \frac{1}{u(0)} \strokedint_{-\io}^{\io} \frac{d\tp}{2\pi}
\left| \frac{d}{d\tp} \left( \frac{\tf(\tp) - \tf(0)}{\tp} \right) \right| < \io
\ee
since the conditions on $f(p)$ in \eqref{f_conditions} and $u(p)$ in \eqref{u_conditions} imply that the integrand is $L^1$: this follows from a straightforward computation showing that
\be
\frac{\tf(\tp) - \tf(0)}{\tp}
= \left(
	\frac{f(p) - f(0)}{p} - f(0) \frac{d}{dp} \ln \left( \frac{u(p)}{u(0)} \right) \right) \frac{dp}{d\tp}
\ee
and
\begin{multline}
\frac{d}{d\tp} \left( \frac{\tf(\tp) - \tf(0)}{\tp} \right)
= \frac{d}{dp} \left(
	\frac{f(p) - f(0)}{p} - f(0) \frac{d}{dp} \ln \left( \frac{u(p)}{u(0)} \right)
\right) \left( \frac{dp}{d\tp} \right)^2 \\
+ \left(
	\frac{f(p) - f(0)}{p} - f(0) \frac{d}{dp} \ln \left( \frac{u(p)}{u(0)} \right)
\right) \frac{d^2p}{d\tp^2}
\end{multline}
with ${dp}/{d\tp} = {u(0)}/{u^{\textnormal{g}}(p)}$ and ${d^2p}/{d\tp^2} = {d}( {u(0)^2}/{2u^{\textnormal{g}}(p)^2} )/{dp}$.
(The latter contains the second derivative of $u(p)$, which explains why the condition on $d^2u(p)/dp^2$ in \eqref{u_conditions} is needed.)
\end{proof}


\section{General interactions}
\label{App:General_interactions}
We consider an equilibrium model of the same form as the one in Sec.~\ref{Sec:Luttinger_model_with_constant_chemical_potentials} but for general $g_2$- and $g_4$-interactions in \emph{g-ology}; see, e.g., \cite{V}.
This model has different chemical potentials for right- and left-moving fermions and $\l\Hint$ is decomposed into $g_2 H_{2}$ (interaction between fermions moving in opposite directions) and $g_4 H_{4}$ (interaction between fermions moving in the same direction):
\be
\begin{aligned}
	H & = H_{0} + g_2 H_{2} + g_4 H_{4} - \sum_{r} (\m_{r} - \m_0) Q_{r}, \\
	H_{2} & = \sum_{r} \frac{1}{L} \hat{\Vint}_2(0) Q_{r} Q_{-r}
		+ \sum_{r} \sum_{p>0} \frac{1}{L} \hat{\Vint}_2(p) \left(
			\r_{r}(p)\r_{-r}(-p) + \r_{r}(-p)\r_{-r}(p)
		\right), \\
	H_{4} & = \sum_{r} \frac{1}{L} \hat{\Vint}_4(0) Q_{r}^2
		+ \sum_{r} \sum_{p>0} \frac{1}{L} \hat{\Vint}_4(p) \left(
			\r_{r}(p)\r_{r}(-p) + \r_{r}(-p)\r_{r}(p)
		\right)
\end{aligned}
\label{general_hamiltonian}
\ee
with $\m_\pm$ constant chemical potentials satisfying $\m_+ + \m_- = 2\m_0$; the conditions on the $g_2$- and $g_4$-interactions can be shown to be $\hat{\Vint}_i(p) = \hat{\Vint}_i(-p)$ for $i=2,4$, $|g_2 \hat{\Vint}_2(p)| < \pi\vF + g_4 \hat{\Vint}_4(p)$ for all $p$, and
\be
\sum_{p > 0}
\frac{p g_2^2 \hat{\Vint}_2(p)^2}{\pi\vF \left( \pi\vF + g_4\hat{\Vint}_4(p) \right)} < \io.
\ee
We also find it convenient to introduce the following shorthand notation:
\be
\tg_2(p) = g_2 \hat{\Vint}_2(p)/\pi\vF,
\qquad
\tg_4(p) = g_4 \hat{\Vint}_4(p)/\pi\vF,
\qquad
\tm_{r} = (\m_{r} - \m_0)/\pi\vF.
\ee
The renormalized Fermi velocity and Luttinger parameter can then be written
\be
\begin{aligned}
v_{g_2,g_4}(p) & = \vF \sqrt{(1+\tg_4(p))^2 - \tg_2(p)^2}, \\
K_{g_2,g_4}(p) & = \sqrt{({1 + \tg_4(p) - \tg_2(p)})/({1 + \tg_4(p) + \tg_2(p)})},
\end{aligned}
\label{vL_and_Luttinger_parameter_for_g2_g4}
\ee
respectively, for general $g_2$- and $g_4$-interactions \cite{V}.
Moreover, similar to Sec.~\ref{Sec:Luttinger_model_with_constant_chemical_potentials}, the possible values of $\m_\pm - \m_0$ must be integer multiples of an interaction dependent constant (see \eqref{general_w}), but as before this will be of no consequence in the thermodynamic limit.

We define a current which is consistent with the continuity equation implied by $H$.
Using $\r_{r}(p,t) = e^{iH_{\l} t} \r_{r}(p) e^{-iH t}$ we have $\partial_t \r_{r}(p,t) = ie^{iH t} [H, \r_{r}(p)] e^{-iH t}$, where $[H_{0}, \r_{r}(p)]$ is given in \eqref{H0_rsp}, $[H_{2}, \r_{r}(p)] = - r\vF p (\hat{\Vint}_2(p)/\pi\vF) \r_{-r}(p)$, and $[H_{4}, \r_{r}(p)] = - r\vF p (\hat{\Vint}_4(p)/\pi\vF) \r_{r}(p)$.
It follows that
\be
\partial_t \left( \r_+(p,t) + \r_-(p,t) \right)
	+ ip \vF
		\left( 1 + \tg_4(p) - \tg_2(p) \right)
		\left( \r_+(p,t) - \r_-(p,t) \right) = 0.
\ee
This is the continuity equation in Fourier space if we define the total density and the current as $\r(p) = \r_+(p) + \r_-(p)$ and
\be
j(p)
= \vF \left( 1 + \tg_4(p) - \tg_2(p) \right) (\r_+(p) - \r_-(p))
= K_{g_2,g_4}(p) v_{g_2,g_4}(p) (\r_+(p) - \r_-(p)),
\label{general_current}
\ee
respectively.
Setting $g_2 = g_4 = \l$ and $\hat{\Vint}_2(p) = \hat{\Vint}_4(p) = \hat{\Vint}(p)$, we recover $\l \Hint = g_2 H_{2} + g_4 H_{4}$ and $j(x) = \vF (\r_+(x) - \r_-(x))$ since $v_{\l}(p)$ in \eqref{vL} and  $K_{\l}(p)$ in \eqref{Luttinger_parameter} satisfy $K_{\l}(p) v_{\l}(p) = \vF$.
This proves that the total density and the current in the main text are consistent with the corresponding continuity equation.

Let $|\Psi\rangle$ denote the ground state of $H$ and consider the average current
\be
I
= L^{-1} \int_{-L/2}^{L/2} dx\, \langle \Psi| j(x) |\Psi \rangle
= L^{-1} \langle \Psi| j(p = 0) |\Psi \rangle,
\ee
which corresponds to the steady current in the main text.
Only the zero modes contribute to this current since
\be
I = \frac{K_{g_2,g_4}v_{g_2,g_4}}{L} \langle\Psi| \left( Q_+ - Q_- \right) |\Psi\rangle,
\label{asymptotic_model_current}
\ee
where $v_{g_2,g_4} = v_{g_2,g_4}(0)$ and $K_{g_2,g_4} = K_{g_2,g_4}(0)$, using $j(p)$ in \eqref{general_current}.
Similar to Sec.~\ref{Sec:Luttinger_model_with_constant_chemical_potentials}, the Hamiltonian $H$ can be diagonalized using a unitary operator $e^{iS_{g_2,g_4}}$, given by
\be
S_{g_2,g_4} = i \sum_{p \neq 0} \frac{2\pi}{L} \frac{\ph_{g_2,g_4}(p)}{p} \r_{+}(-p)\r_-(p),
\qquad
\tanh 2\ph_{g_2,g_4}(p) = -\frac{\tg_2(p)}{1 + \tg_4(p)}
\quad
\forall p \neq 0,
\label{S_g2_g4}
\ee
and large gauge transformations implemented by the unitary operators $R_\pm$ defined in Sec.~\ref{Sec:Exact_solution_of_the_Luttinger_model}.
It follows that the ground state of $H$ can be written as $|\Psi\rangle = R_+^{w_+} R_-^{-w_-} |\Psi_{g_2,g_4}\rangle$ with $|\Psi_{g_2,g_4}\rangle = e^{-iS_{g_2,g_4}} |\Psi_0\rangle$.
Since we are interested in computing $\langle\Psi| \left( Q_+ - Q_- \right) |\Psi\rangle = w_+ - w_-$, it suffices to consider the zero-mode part of $|\Psi\rangle$.
We therefore consider only the zero-mode part of $H$,
\be
H_{Q}
= \sum_{r} \frac{\pi\vF}{L}
	\left( \left( 1 + \tg_4(0) \right) Q_{r}^2 + \tg_2(0) Q_{r}Q_{-r} - L\tm_{r} Q_{r} \right).
\ee
The large gauge transformations are defined so that they remove all terms which are linear in the charge operators $Q_\pm$,
\begin{multline}
R_+^{-w_+} R_-^{w_-} H_{Q} R_+^{w_+} R_-^{-w_-}
= H_{Q}
- \sum_{r} \frac{\pi\vF}{L}
	\left(
		\left( 1 + \tg_4(0) \right) w_{r} + \tg_2(0) w_{-r}
	\right) w_{r} \\
+ 2 \sum_{r} \frac{\pi\vF}{L}
	\left(
		\left( 1 + \tg_4(0) \right) w_{r} + \tg_2(0) w_{-r} - L\tm_{r}/2
	\right) (Q_{r} + w_{r}),
\end{multline}
which, using $\m_+ + \m_- = 2\m_0$, implies that $w_\pm$ must be chosen as
\be
w_\pm
= \frac{L\tm_\pm/2}{1 + \tg_4(0) - \tg_2(0)}
= \frac{L(\m_{r} - \m_0)}{2\pi K_{g_2,g_4}v_{g_2,g_4}}.
\label{general_w}
\ee
From this and \eqref{asymptotic_model_current}, using $\langle\Psi| \left( Q_+ - Q_- \right) |\Psi\rangle = w_+ - w_-$, we find that $K_{g_2,g_4}v_{g_2,g_4}$ cancels:
\be
I = \frac{K_{g_2,g_4}v_{g_2,g_4}}{L} \frac{L(\m_+ - \m_-)}{2\pi K_{g_2,g_4}v_{g_2,g_4}}
	= \frac{\m_+ - \m_-}{2\pi}.
\ee
This corresponds to a universal conductance $G = {1}/{2\pi}$, i.e., the conductance quantum found in the main text; the computation given here is essentially an alternative derivation of the result in \cite{ACF2}, but for a less general family of Hamiltonians.

We also show that it is possible to read the chemical potential difference from the two-point correlation function, as done in the main text.
Indeed, from \eqref{large_gauge_transf}, it follows that
\be
\langle\Psi| \psi^+_{r}(x) \psi^-_{r}(y) |\Psi\rangle
= e^{\mp ir w_{r} 2\pi (x-y)/L}
	\langle\Psi_{g_2,g_4}| \psi^+_{r}(x) \psi^-_{r}(y) |\Psi_{g_2,g_4}\rangle,
\ee
which, using \eqref{general_w}, implies that the two-point correlation function is
\be
\langle\Psi| \psi^+_{r}(x) \psi^-_{r}(y) |\Psi\rangle
= e^{-ir(K_{g_2,g_4}v_{g_2,g_4})^{-1} (\m_{r}-\m_0) (x-y)}
	\langle\Psi_{g_2,g_4}| \psi^+_{r}(x) \psi^-_{r}(y) |\Psi_{g_2,g_4}\rangle.
\ee
The latter is of the same form as \eqref{psi_psi_correl_2} but with the factor $(K_{g_2,g_4}v_{g_2,g_4})^{-1}$ replacing $\vF^{-1}$ in the phase on the right-hand side.
In particular, setting $g_2 = g_4 = \l$ and $\hat{\Vint}_2(p) = \hat{\Vint}_4(p) = \hat{\Vint}(p)$ we recall that $K_{\l}v_{\l} = \vF$, while, in general, the new factor cancels the factor $K_{g_2,g_4}v_{g_2,g_4}$ in the current (given by \eqref{general_current} for $p = 0$).
This shows that universality of conductance holds for general $g_2$- and $g_4$-interactions and not only for the particular interaction in the main text.

As a final remark, we mention that the formulas given in the main text for the evolution of the Luttinger model following a quench can be shown to remain valid for general $g_2$- and $g_4$-interactions if one replaces $v_{\l}(p)$ and $K_{\l}(p)$ with $v_{g_2,g_4}(p)$ and $K_{g_2,g_4}(p)$ defined in \eqref{vL_and_Luttinger_parameter_for_g2_g4} (the other necessary changes are explained above).
In particular, we note that the relation in \eqref{Luttinger_parameter_ph} holds replacing $\ph_{\l}(p)$ with $\ph_{g_2,g_4}(p)$ in \eqref{S_g2_g4} but that the relation in \eqref{vL_ph} does \emph{not} generalize.


\end{appendix}



\end{document}